\documentclass[11pt]{article}
\usepackage[top=.9in,bottom=.9in,right=.75in,left=.75in]{geometry}
\usepackage{amssymb,amsmath,amsthm}
\usepackage{tikz}
\usepackage{arydshln}
\usepackage{verbatim}
\usepackage{titlesec}
\usepackage{bbm}
\usepackage{bm}
\usepackage{subfigure}
\usepackage{ stmaryrd }
\usepackage[linesnumbered,ruled]{algorithm2e}
\usepackage{titling}
\definecolor{forestgreen}{rgb}{0.13, 0.55, 0.13}
\usepackage[colorlinks, linkcolor=blue,citecolor=blue,urlcolor = black, bookmarks=true]{hyperref}
\usepackage[nameinlink, noabbrev, capitalize]{cleveref}
\usepackage{xcolor}

\usepackage[suppress]{color-edits}
\addauthor{et}{purple}
\addauthor{jg}{cyan}
\newcommand{\Xomit}[1]{}

\titleclass{\subsubsubsection}{straight}[\subsection]

\newcounter{subsubsubsection}[subsubsection]
\renewcommand\thesubsubsubsection{\thesubsubsection.\arabic{subsubsubsection}}

\titleformat{\subsubsubsection}
  {\normalfont\normalsize\bfseries}{\thesubsubsubsection}{1em}{}
  \titlespacing*{\subsubsubsection}
  {0pt}{3.25ex plus 1ex minus .2ex}{1.5ex plus .2ex}

\makeatletter  
\def\toclevel@subsubsubsection{4}
\def\l@subsubsubsection{\@dottedtocline{4}{7em}{4em}}

\makeatother
\crefname{equation}{}{}
\crefrangeformat{equation}{(#3#1#4) --~(#5#2#6)}
\crefmultiformat{equation}{(#2#1#3)}{, (#2#1#3)}{, (#2#1#3)}{, (#2#1#3)}
\crefname{lem}{Lemma}{Lemmas}
\crefname{section}{Section}{Sections}
\crefname{subsubsubsection}{Section}{Sections}
\crefname{rem}{Remark}{Remarks}
\crefname{figure}{Figure}{Figures}
\crefname{table}{Table}{Tables}
\Crefname{lem}{Lemma}{Lemmas}
\crefname{thm}{Theorem}{Theorems}
\Crefname{thm}{Theorem}{Theorems}

\newtheorem{thm}{Theorem}[section]

\newtheorem{example}{Example}[section]

\newtheorem{lem}{Lemma}[section]

\newtheorem{proposition}[thm]{Proposition}
\newtheorem{corollary}[thm]{Corollary}

\theoremstyle{definition}
\newtheorem{fact}{Fact}[section]
 
\theoremstyle{definition}
\newtheorem{defn}{Definition}[section]

\title{Virtues of Patience in Strategic Queuing Systems}
\author{Jason Gaitonde\thanks{Supported by NSF grant CCF-1408673 and AFOSR grant FA9550-19-1-0183.} \\
Cornell University\\
\texttt{jsg355@cornell.edu}
\and 
\'Eva Tardos\thanks{Supported in part by NSF grant CCF-1408673 and AFOSR grant FA9550-19-1-0183.} \\
Cornell University\\
\texttt{eva.tardos@cornell.edu}
}

\begin{document}
\maketitle
\thispagestyle{empty}

 \begin{abstract}
We consider the problem of selfish agents in discrete-time queuing systems, where competitive queues try to get their packets served. In this model, a queue gets to send a packet each step to one of the servers, which will attempt to serve the oldest arriving packet, and unprocessed packets are returned to each queue. We model this as a repeated game where queues compete for the capacity of the servers, but where the state of the game evolves as the length of each queue varies, resulting in a highly dependent random process. In classical work for learning in repeated games, the learners evaluate the outcome of their strategy in each step---in our context, this means that queues estimate
their success probability at each server. Earlier work by the authors [in EC'20] 
shows that with no-regret learners the system needs twice the capacity as would be
required in the coordinated setting to ensure queue lengths remain stable, despite the selfish behavior of the queues. In this paper, we demonstrate that this myopic way of evaluating outcomes is suboptimal: if more patient queues choose strategies that selfishly maximize their \emph{long-run success rate}, stability can be ensured with just $\frac{e}{e-1}\approx 1.58$ times extra capacity, strictly better than what is possible assuming the no-regret property.

 As these systems induce highly dependent random processes, our analysis draws heavily on techniques from the theory of stochastic processes to establish various game-theoretic properties of these systems. Though these systems are random even under 
 fixed stationary policies by the queues, we show using careful probabilistic arguments that surprisingly, under such fixed policies, these systems have essentially deterministic and explicit asymptotic behavior. We show that the growth rate of a set can be written as the ratio of a submodular and modular function, and
 use the resulting explicit description to show that the subsets of queues with largest growth rate is closed under union and non-disjoint intersections, which we use in turn 
 to prove the claimed sharp bicriteria result for the equilibria of the resulting system. Our equilibrium analysis relies on a  
 novel deformation argument towards a more analyzable
 solution that is quite different from classical price of anarchy bounds. While the intermediate points in this deformation will \emph{not} be Nash, the structure will ensure the relevant constraints and incentives similarly hold to establish monotonicity along this continuous path.
\end{abstract}

\newpage
\setcounter{page}{1}
\section{Introduction}

A fundamental aim at the intersection of economics and computer science is to understand the efficiency of systems when the dynamics are governed by the actions of strategic and competitive agents. 
A large body of work has established bounds on the \emph{price of anarchy} of various well-studied games, which measures the gap between the social welfare at the worst-case Nash equilibrium with that of the social optimum attainable via coordination by the agents \cite{DBLP:conf/stacs/KoutsoupiasP99}; moreover, recent work has shown that in many cases, price of anarchy bounds often seamlessly extend when agents employ simple no-regret learning algorithms in repeated games \cite{DBLP:conf/stoc/BlumHLR08,DBLP:journals/jacm/Roughgarden15,DBLP:conf/soda/LykourisST16}.

However, a critical assumption of these models is that in each round, agents play an ``independent'' copy of the \emph{same} game; that is, the past sequence of play does not fundamentally change the nature of the game except through their learning. In many settings, this approximation might be innocuous and valid, for instance in modeling routing in the context of the morning rush-hour traffic. In other applications, this assumption clearly does not hold; in modeling packet routing in computer networks, if a packet gets dropped, then this packet will be resent
in future rounds and increase congestion. Therefore, developing a deeper theory of the efficiency of strategic agents in repeated games that retain state is of large practical importance.

Recent work by the authors \cite{DBLP:conf/sigecom/GaitondeT20}, to our knowledge, initiates the study of price-of-anarchy-style bounds when the repeated games carry state in a discrete-time queuing system. In this model, queues must compete for servers to clear their own packets which arrive at heterogeneous rates and servers give priority to serving older packets. This priority scheme induces strong dependencies between rounds, as a queue that fails to clear packets will have priority in future rounds. Their main result is that if queues use no-regret algorithms, then just a factor of 2 higher service rates than what is needed for centralized stability is sufficient to ensure selfish queues also remain stable.

In this paper, we show
that no-regret behavior can nonetheless look short-sighted when considering the long-run dependencies at play in this queuing system, as seen in \Cref{ex:ex1}. While no-regret learning is well-suited to ``independent'' repeated games, strategies satisfying no-regret need not be so well-adapted when adding state. 
We explore the outcomes of repeated games with state in this queuing setting 
when queues are sufficiently \emph{patient} in evaluating the results of their strategies and select fixed randomizations over servers optimally conditioned on the policies of the other queues. We connect the \emph{game-theoretic incentives} of this system with the \emph{asymptotic properties of these random dynamics} by enforcing that each queue aims to minimize their linear rate of growth under these dynamics. Using a novel deformation argument, we show that just $\frac{e}{e-1}< 2$ times extra server capacity suffices to ensure every queue is stable in any Nash equilibrium of this game, strictly better than what is possible under no-regret dynamics.\\

 \subsection{Overview of Results and Techniques}
 In this paper, we consider a discrete-time queuing system where $n$ queues receive packets at heterogeneous rates $\bm{\lambda}=(\lambda_1,\ldots,\lambda_n)\in (0,1)^n$. In each round, any queue that has any remaining packets must select exactly one of $m$ servers with heterogeneous success probabilities $\bm{\mu}=(\mu_1,\ldots,\mu_m)\in [0,1]^{m}$, to attempt to clear a single packet. Each server can only succeed in clearing at most one packet in each round, and most importantly, returns each unprocessed packet to the original queue,  assuming for simplicity that servers have no buffer. We assume that servers choose to serve the \emph{oldest packet} it receives in each round, thus giving priority to queues that have not been able to clear packets efficiently.\footnote{An alternate choice is to assume that servers select packets uniformly at random among packets it receives in a given round to attempt to serve. However, it was shown in \cite{DBLP:conf/sigecom/GaitondeT20} that adding even $\mathrm{poly}(n)$ slack to the central feasibility conditions (given by (\ref{eq:centfeas})) cannot guarantee stability in general, and patience by the queues does not help in this case either.} 
 Queue lengths can grow arbitrarily, so the efficiency we consider is under what conditions on the service and arrival rates can it be ensured that the system remain stable? See  \Cref{section:standardsystem} for precise definitions and the full specification. 
 
 The main result of this paper is to show that with patience by the queues, a factor of $\frac{e}{e-1}\approx 1.58$ extra server capacity \jgedit{over what is needed in the centralized setting} suffices to guarantee the stability of the system, despite selfishness. This is in contrast to \jgedit{the main result of} \cite{DBLP:conf/sigecom/GaitondeT20}, where it was shown that if queues use no-regret learning, the system needs a factor of $2$ extra server capacity. \jgedit{Their analysis derives this factor by leveraging the no-regret property along with concentration and a potential function argument to ensure that older queues use the most efficient servers sufficiently many times on long enough windows, or else some such queue must experience regret; the factor $2$ arises in a natural way from these considerations, and can be shown to be essentially unimprovable just assuming the no-regret property}. \jgedit{While their result provides a constant factor bicriteria result}, \Cref{ex:ex1} illustrates the value of patience in our model. In this example, there is a \emph{unique} fixed no-regret policy for each agent, 
given the behavior of the other agent, 
\emph{but both agents would have been better off long-term had one even slightly deviated to an inferior server and the other stayed the same}. In the classical setting with ``independent'' repeated games, this cannot occur. 
\begin{example}
\label{ex:ex1}
Suppose there are two queues with arrival rates $\lambda_1=\lambda_2=.51$ and two servers with $\mu_1=1$ and $\mu_2=.49$. In this case, each queue receives a new packet roughly once every two periods on average. One can show that if both servers send to the top rate server every period, the sequence of play will satisfy the no-regret property, as they roughly split the top server equally. Each server then roughly clears at a rate of $1/2$ which is strictly better than deviating to the lower server, but this system will not be stable; packets arrive at a total rate of $1.02$ while are cleared at a rate of $1$ in expectation, leading to linear growth. However, if one queue commits to slightly deviate towards the inferior server, one can show that the resulting system will be stable. 
\end{example}

No-regret algorithms like EXP3 empirically seem to exhibit similar behavior: the queue sending to the second server does have regret, despite doing better in the long term. However, if queues are viewed as sufficiently patient and could choose \emph{fixed} strategies that optimize long-run stability by sometimes sending to suboptimal servers, they could experience better long-run stability. In the above example, this happens because if one queue starts mixing slightly on the slow server, then the other queue clears faster as she faces less competition on the high rate server. In turn, although the first queue sends to the fast server less often, she will tend to have priority more often when she does do so, so the long-run effectiveness of this server actually increases for this agent. These effects cause both servers to simultaneously clear, as we will see.
Though such deviations look suboptimal, if the deviating queue is sufficiently patient, she will see that this change yields sharp asymptotic benefits.

\subsubsection{Patient Selfishness}
While vanilla no-regret learning seemed to be the ``correct'' notion to study repeated games without carryover, examples like this suggest that perhaps this is not so when \jgedit{outcomes from previous rounds of the game directly affect the nature of the games in future iterations}. 
In this work, we formulate and study a \emph{patient} version of this model where queues optimize over the long term effect of \emph{fixed} randomized strategies over servers. Each queue $i$'s choice over servers 
can be described by a fixed vector $p_i\in \Delta^{m-1}$, where $\Delta^{m-1}$ is the probability simplex over the $m$ servers. We study this as a traditional game and consider the resulting Nash equilibria when each queue aims to choose their fixed randomization to minimize their \emph{long-run aging rate} (equivalently, their long-run growth rate, see \Cref{section:standardsystem}) conditioned on the others. Our main interest is understanding under what conditions on the service rates and arrival rates will the system remain stable in \emph{every Nash equilibrium}? To study this, we face significant probabilistic and game-theoretic challenges: probabilistic challenges to establish the close form of asymptotic growth rates for given strategies, and game-theoretic challenges bounding the quality of Nash equilibria of this game. The techniques we use will prove useful in addressing these conceptually distinct difficulties, and thereby unifying the game-theoretic and probabilistic properties of our systems.

\begin{enumerate}

\item\textbf{Asymptotic Growth Rates:} In the above discussion, we stated that each queue aims to select a fixed randomization over servers to minimize their long-run aging rate in this system given the randomizations of the others. Our first task, to do any game-theoretic analysis of this system, is to analyze the long-run properties of this random process of queue ages (which typically will not even be recurrent). A major technical component of our work is showing that for any fixed, independent randomizations $\mathbf{p}$ by the queues over servers, not only do these long-run growth rates exist almost surely, they are \emph{deterministic} and can be explicitly computed as a function of the strategies.
    
    \begin{thm}[\Cref{thm:gamemain}, informal]
    There exists an explicit, continuous function $r:(\Delta^{m-1})^n\to \mathbb{R}^n_{\geq 0}$ such that, if queues independently randomize over servers according to $\mathbf{p}\in (\Delta^{m-1})^n$, then the (random) long-run growth rate of each queue $i$ is $r_{i}(\mathbf{p})$ almost surely.
    \end{thm}
    
\noindent   To prove this result, we provide an alternate, \emph{algorithmic} description of the long-run rates in \Cref{sec:algdescription}. Working just with this alternative definition, we show that the queues partition into 
groups such that all queues in a group age asymptotically at the same rate. We then return to the task of establishing that the true, long-run asymptotic aging rates of the queues for any choice of strategies coincides with the output of the algorithm. To do this, we repeatedly appeal to concentration bounds to show 
that each subset in the partition grows at the desired rate; as the priority structure changes rapidly round-to-round, we do so via a delicate argument that accounts for these changes. We then carefully apply the Borel-Cantelli lemma to establish the result. After more formally defining the various parameters of our queuing process, we prove 
this result in \Cref{section:convergence}, but defer some 
of the highly nontrivial and technical details to \cref{sec:appendixconv}.
    
\item\textbf{Game-Theoretic Properties: Equilibria and Price of Anarchy.} Once we show that these limits almost surely are equal to an explicit, 
deterministic function of $\mathbf{p}$, it might still not be the case that a Nash equilibrium exists in the induced game. However, we show that the cost function exhibits significant analytic properties which lets us reason about the structure of the sets that arise in the partition for any fixed strategy profile. Concretely, we show that each level set of the cost function corresponds to the minimizing subset of 
the ratio of a submodular and modular set function; this significant structure allows us to show that the subsets that minimize this ratio are closed under union and non-empty intersection. In particular, these considerations will be enough to show continuity as a function of the strategies (\Cref{thm:ratecontinuity}) 
which along with other properties will enable us to show that an equilibrium exists using Kakutani's Theorem
(\Cref{thm:nashexists}). While we show that the cost function of our game has significant structure, the correspondence between actions (randomizations) and costs is quite nonlinear, imposing new technical challenges.

Recall that our goal is to ensure stability in any Nash equilibrium, assuming some relationship on the service rates to the arrival rates. Our main result shows
    that the correct constant of system slack is $\frac{e}{e-1}\approx 1.58$,  beating the best achievable constant of $2$ in the no-regret setting of \cite{DBLP:conf/sigecom/GaitondeT20}: 
    \begin{thm}[Main, \Cref{thm:ebound}, informal]
    \label{thm:intromain}
    If the service capacity is large enough so that the system would remain feasible when centrally managed even if capacities are scaled down by $\frac{e}{e-1}$, then in every equilibrium of this game, all queues are stable.
    \end{thm}
\noindent    This result is tight---in a symmetric system where $n=m$, each queue has the same arrival rate, each server has the same success rate, and each queue chooses to uniformly randomize over servers, a simple balls-in-bins analysis yields this constant as $m,n\to \infty$.

    To prove this theorem, we provide a novel argument that establishes the result by continuously deforming any Nash profile towards a carefully constructed strategy profile, while only monotonically decreasing the rate at which the top group clears. We then analyze the resulting profile to give a lower bound on the value of the Nash profile. The key difficulty is that the relevant incentives for each queue correspond to possibly \emph{many different} subsets of queues that have maximal aging rate, subsets that collectively clear packets at the lowest rate relative to the rate they receive them. These constraints are difficult to directly compare; different choices of deviations in the strategy by a queue at any Nash equilibrium may violate distinct constraints, making it unclear how to argue about the quality of these equilibria. In particular, there does not seem to be a direct analogue of the \emph{Nash indifference principle} in finite-action games where utilities are affine in the randomizations of each agent (recall \Cref{ex:ex1}, where the queue moving to the lesser server will still appear to prefer the better server). To overcome these difficulties, we show that one can significantly reduce the number of incentive constraints one must consider for each queue (\Cref{prop:levels}). We can carefully perform our deformation of the collective strategy vector of the queues according to the structure of these sparsified constraints, and show that our deformation only hurts the quality of the Nash solution to provide a valid lower bound. We elaborate on these difficulties and prove \Cref{thm:intromain} in \Cref{section:poa}.
    
    In contrast, almost every known price-of-anarchy-style result can be viewed via the very general \emph{smoothness} framework of Roughgarden \cite{DBLP:journals/jacm/Roughgarden15}, which connects an equilibrium with the social optimum via discrete changes in the strategy profile. Our argument instead relies on a careful equilibrium analysis that smoothly interpolates between the equilibrium and a ``good'' profile which is easy to explicitly bound; however, during these deformations, \emph{these intermediate strategy profiles will not be equilibria}. To prove the monotonicity of this deformation, we connect the incentives at Nash to the structure of the subset of maximizers of the long-run rate function, and show that the Nash constraints still hold in the directions we deform.
    \end{enumerate}
    
   The organization of this paper is as follows: after formalizing the strategic queuing model in \Cref{section:standardsystem} and the relevant game in \Cref{section:patient}, we describe how to compute the resulting aging rates for each fixed strategy $\mathbf{p}$ by the queues in \Cref{sec:algdescription}. Roughly speaking, given any fixed $\mathbf{p}$, the set of queues partitions into subsets that all age at equal asymptotic rates determined by fractional bottlenecks in the system. 
   Using this algorithmic description, we will then turn to establishing analytic properties in  \Cref{section:basicproperties}; from these properties, we will be able to deduce the existence of pure equilibria in this queuing game in  \Cref{section:equilibria}. We then provide the proof that the algorithmic description of long-run queuing rates given in the previous sections indeed holds 
   with probability one in \Cref{section:convergence}; as the details are rather involved, however, we defer some of the more technical auxilliary claims to \Cref{sec:appendixconv}. Finally, we establish our tight bound on the price of anarchy of $\frac{e}{e-1}$ in \Cref{section:poa}.

\subsection{Related Work}
Our work falls in a long tradition of establishing \emph{price of anarchy} bounds for various games \cite{DBLP:conf/stacs/KoutsoupiasP99}, which roughly quantifies the difference in the performance of a competitive and selfish system with the social optimum that can be achieved through explicit coordination. As mentioned, our work is most closely related to the prior work in \cite{DBLP:conf/sigecom/GaitondeT20}. While our stability objective differs from usual objectives in this literature, our results qualitatively also resemble the bicriteria result of Roughgarden and Tardos \cite{DBLP:journals/jacm/RoughgardenT02}, which shows that in nonatomic routing, 
the cost incurred at any Nash flow is at most the optimal cost when twice the flow is routed. 
Unlike most such bounds which follow Roughgarden's \emph{smoothness} framework \cite{DBLP:journals/jacm/Roughgarden15}, our argument is an \emph{equilibrium analysis} that is more similar 
to Johari and Tsitsiklis \cite{DBLP:journals/mor/JohariT04}, who establish equilibrium conditions and modify their problem while maintaining the equilibrium condition to arrive at a version that is easy to analyze. In our argument, we also modify the equilibrium itself 
towards a more tractable solution, but 
the intermediate points in this deformation will \emph{not} be Nash, requiring additional arguments. 

While the goal of our work is in establishing price-of-anarchy-style bounds in dependent systems, this necessitates a careful understanding of the analytic properties of our random queuing dynamics. In the probability literature, the long-run aging rates that form the \emph{incentives} are known as (linear) escape rates of random walks \cite{MR3616205}. A rich theory has emerged to study this in special networks, but it is unclear how to apply these techniques in our setting. Our work relies on careful, self-contained estimates leveraging concentration, coupling, and supermartingale arguments.

Our queuing setting bears resemblance to \emph{stochastic games}, a generalization of Markov decision processes where multiple players competitively and jointly control the actions and transitions (see, for instance \cite{filar2012competitive,neyman2003stochastic}). However, our work differs from this line in multiple ways: in our model, queues are unaware of the system state and parameters, and most importantly, we are interested in explicit bounds to derive price-of-anarchy-style results for stability. 

\section{Preliminaries}
\textbf{Notation.} We use the following \textbf{fractional sum} operation $\oplus: \mathbb{R}^2\times \mathbb{R}^2_{\geq 0}\to \mathbb{R}_{\geq 0}$:
\begin{equation*}
    \frac{a}{b}\oplus \frac{c}{d}\triangleq \frac{a+c}{b+d}.
\end{equation*}
We will later repeatedly use the following simple fact:
\begin{fact}
\label{fact:dansineq}
For all $a_1,\ldots,a_n\geq 0$ and $b_1,\ldots,b_n> 0$,
\begin{equation*}
    \min_{i\in [n]}\frac{a_i}{b_i}\leq \frac{a_1}{b_1}\oplus \ldots\oplus \frac{a_n}{b_n}=\frac{\sum_{i=1}^n a_i}{\sum_{i=1}^n b_i}\leq \max_{i\in [n]}\frac{a_i}{b_i}.
\end{equation*}
Moreover, equality holds in either of the inequalities if and only if both inequalities are tight.
\end{fact}

Given a $n\cdot m$-dimensional vector $\mathbf{p}=(p_1,\ldots,p_n)$, where $p_i\in \mathbb{R}^m$, we will write $p_{ij}$ for the $j$th element of $p_i$. Given a vector $\mathbf{x}\in \mathbb{R}^n$ and a subset $I\subseteq [n]$, we write $\mathbf{x}_I$ to denote the vector restricted to the components in $I$. As is standard, we say ``almost sure" to mean with probability one. Given a set $S$, we will write $\mathcal{P}(S)$ to denote the power set. We write $\text{Bern}(\lambda)$ to denote a Bernoulli random variable that is $1$ with probability $\lambda$ and $0$ with probability $1-\lambda$. We write $\text{Geom}(\lambda)$ to denote a geometric random variable with parameter $\lambda$.\\

\subsection{Queuing Model}
\label{section:standardsystem}
 We study the strategic queuing model introduced by the authors \cite{DBLP:conf/sigecom/GaitondeT20}, which is a competitive version of a queuing model considered by Krishnasamy et al \cite{DBLP:conf/nips/KrishnasamySJS16}. As described above, queues receive packets with some fixed probability, and must select a server with heterogeneous success rates to send their oldest packet to. Each server chooses only the \emph{oldest} packet it receives to attempt to clear, and returns each unprocessed packet to the original queue (as well as the packet it attempted to clear if it fails). Each queue receives only bandit feedback of whether their packet was cleared or not. In this work, we instead work with an equivalent, \emph{deferred-decisions} version of this model that keeps track only of the oldest packet at each queue:
\begin{enumerate}
    \item Time progresses in discrete steps $t=0,1,\ldots$. At each time $t$, $\widetilde{T}^i_t$ is the \textbf{timestamp} of the oldest unprocessed packet of queue $i$ at time $t$.
    $T^i_t=\max\{0,t-\widetilde{T}^i_t\}$ is the \textbf{age} of the current oldest packet of queue $i$ at time step $t$. That is, $T^i_t$ measures how old the current oldest unprocessed packet for queue $i$ is.\footnote{Note that while $T^i_t\geq 0$ by definition, it is possible that $\widetilde{T}^i_t>t$. The interpretation of this is that the queue has cleared all of her packets at time $t$ and will receive her next one at time $t=\widetilde{T}^i_t$, or equivalently, in $\widetilde{T}^i_t-t$ steps in the future from the perspective at time $t$.}
    \item Queue $i$ can send a packet to any server $j$ in this time step if $t-\widetilde{T}^i_t\geq 0$. Each server $j$ attempts to serve only the oldest packet it receives, and succeeds with probability $\mu_j$. If queue $i$'s packet is successfully served, set $\widetilde{T}^{i}_{t+1}= \widetilde{T}^i_t+X^i$, 
    where $X^i\sim \text{Geom}(\lambda_i)$ is independent of all past events, and otherwise $\widetilde{T}^{i}_{t+1}=\widetilde{T}^{i}_{t}$.
\end{enumerate}
See \cite{DBLP:conf/sigecom/GaitondeT20} for a more formal explanation of the equivalence: it follows because the gap between a sequence of $\text{Bern}(\lambda)$ trials follows a $\text{Geom}(\lambda)$ distribution.
\begin{defn}
\label{defn:dualstable}
The queuing system under some dynamics is \textbf{stable} if for each $i\in [n]$, $T^i_t/t\to 0$ almost surely. The  system is \textbf{strongly stable} if, for any $r\geq 0$ and any $i\in [n]$, the random process $\{T_t^i\}^{\infty}_{t=0}$ satisfies $\mathbb{E}[(T^i_t)^r]\leq C_r$ where $C_r$ is a fixed constant depending only on $r$, not on $t$.
\end{defn}

\noindent It is shown in \cite{DBLP:conf/sigecom/GaitondeT20} that strong stability is indeed stronger than stability, for here it implies that almost surely $T^i_t =o(t^c)$ for \emph{every} $c>0$. 

To set a baseline measure of when stability is possible under any coordinated policy, it is shown in \cite{DBLP:conf/sigecom/GaitondeT20} that a simple criteria relating $\bm{\lambda}$ and $\bm{\mu}$ is both necessary and sufficient.

\begin{thm}[Theorem 2.1 of \cite{DBLP:conf/sigecom/GaitondeT20}]
\label{thm:feasibility}
Suppose that $1>\lambda_1\geq \ldots\geq \lambda_n>0$ and $1\geq \mu_1\geq\ldots\geq \mu_m\geq 0$.\footnote{The assumption that $\lambda<1$ is merely to avoid simple edge cases that can be separately handled easily.} 
Then the above queuing system is strongly stable for some centralized (coordinated) scheduling policy if and only if for all $1\leq k\leq n$,
\begin{equation}
\label{eq:centfeas}
    \sum_{j=1}^k \mu_j > \sum_{i=1}^k \lambda_i.
\end{equation}
When (\ref{eq:centfeas}) holds, we say that the queuing system is \textbf{(centrally) feasible}.
\end{thm}
Until now, we have mostly left the manner in which queues choose servers unspecified. One natural model is that each queue uses a standard \emph{no-regret learning algorithm} that learns from their previous history of successes at each server to make a (randomized) choice of server in each subsequent round.
In \cite{DBLP:conf/sigecom/GaitondeT20}, it is shown that if the queues 
each use sufficiently good no-regret algorithms, then under mild technical restrictions, 
the system remains stable with just an extra factor of $2$ on the right side of (\ref{eq:centfeas}):

 \begin{thm}[Theorem 3.1 of \cite{DBLP:conf/sigecom/GaitondeT20}, informal]
 \label{thm:gtmain}
 Suppose each queue uses sufficiently good no-regret learning algorithms, and that for all $k\leq n$,
 \begin{equation*}
     \sum_{j=1}^k \mu_j > 2\sum_{i=1}^k \lambda_i
\end{equation*}
 then the random process $\mathbf{T}_t$ under these dynamics is strongly stable. 
 \end{thm}

\subsection{Patient Queuing Systems}
\label{section:patient}
We now formally define the patient queuing game that is the focus of this work.
\begin{defn}
The \textbf{patient queuing game} $\mathcal{G}=([n],(c_i)_{i=1}^n,\bm{\mu},\bm{\lambda})$ is defined as follows: the strategy space for each queue is $\Delta^{m-1}$. Let $\mathbf{p}=(p_1,\ldots,p_n)\in (\Delta^{m-1})^n$ denote the vector of fixed distributions for all queues over servers. The cost function $c_i$ for queue $i$ is
\begin{equation*}
    c_i(p_i,p_{-i}):=\lim_{t\to \infty} \frac{T^i_t}{t},
\end{equation*}
where $T^i_t$ is again the age of queue $i$ at time $t$ in the random queuing process induced by running the queuing system with $\bm{\mu}$ and $\bm{\lambda}$ the system parameters when each queue chooses a server by independently randomizing according to $\mathbf{p}$ in each time step she has a packet. 

Each queue $i$ wishes to minimize $c_i$. We say $\mathbf{p}$ is a \textbf{Nash equilibrium} of $\mathcal{G}$ if for all $i\in [n]$, $p_i\in \arg\min_{p'\in \Delta^{m-1}} c_i(p',p_{-i})$. That is, each queue chooses $p_i$ to minimize their cost function conditioned on the strategies $p_{-i}$ of the other queues.
\end{defn}
As mentioned, we show that this cost function is an explicit, deterministic quantity (see \Cref{thm:gamemain}), from which we will actually show that Nash equilibria exist in \Cref{thm:nashexists}. Our main focus in this work will be to give guarantees 
on the quality of all Nash equilibria. In a slight abuse of the price of anarchy terminology, we make the following definition:

\begin{defn}
Let $\bm{\lambda}$ and $\bm{\mu}$ satisfy the conditions of \Cref{thm:feasibility}, so that $\bm{\mu}$ strictly majorizes $\bm{\lambda}$. For $\alpha\geq 1$, let $\mathcal{G}(\alpha)=([n],(c_i)_{i=1}^n,\bm{\mu},\alpha^{-1}\bm{\lambda})$, and let $\mathcal{N}(\alpha)$ denote the set of Nash equilibria of $\mathcal{G}(\alpha)$. The \textbf{price of anarchy} of $\mathcal{G}=\mathcal{G}(1)$ is defined as the supremum of $\alpha$ values 
such that there exists a Nash equilibrium $\mathbf{p}^*\in \mathcal{N}(\alpha)$ and some $i\in [n]$ such that $c_i(\mathbf{p}^*)>0$ in $\mathcal{G}(\alpha)$.\footnote{In words, the price of anarchy of a centrally feasible system is the supremum of values of $\alpha$ 
such that, when all queue arrival rates are scaled down by $\alpha$, there still nonetheless 
exists a Nash equilibrium and some queue 
that suffers nonzero 
linear aging.} The price of anarchy of the patient queuing game is the supremum over all instances of $\mathcal{G}$.
\end{defn}

\section{Long-Run Strategies}
\label{sec:longrunstrats}

In this section, we extensively study the properties of the cost function $c(\mathbf{p})$, which is currently written as the limit of a random process, which is possibly a random quantity and so of uncertain game-theoretic use. There are immediate difficulties with this: first, the limit need not even exist pathwise, in which case the notion of ``long-run growth rate" itself does not make sense. Just as alarmingly, it could be that the limit does exist almost surely, but \emph{the limit may be a random quantity} and thus of little game-theoretic use.
Finally, even if the limits exist and are equal to a constant almost surely, we require \emph{explicit} knowledge of what these rates are as a function of $\mathbf{p}$ to argue about the existence of equilibria and incentives at an equilibrium.

Our first task is to provide an alternative, algorithmic description of $c(\mathbf{p})$, which we initially denote $r(\mathbf{p})$ (for ``rates'') in \Cref{sec:algdescription} (deferring the proof that this function is equal to the random limit above almost surely to \Cref{section:convergence}); in particular, $c$ will be a deterministic and explicit function of $\mathbf{p}$.
Importantly, we show that $r$ has significant analytic structure. In particular, we show that the level subsets (in $[n]$) of $r(\mathbf{p})$ enjoy convenient closure properties, which will be enough to establish continuity and other properties, which will help to prove the existence of Nash equilibria. We finally return to showing that $c(\mathbf{p})=r(\mathbf{p})$ almost surely, i.e. this alternate description indeed gives the almost sure asymptotic queue aging rates.\\

\subsection{Algorithmic Description of $c$}
\label{sec:algdescription}
We now construct a function $r:(\Delta^{m-1})^n \to [0,1]$ that we will show is equivalent to $c$.
We will show that for any fixed $\mathbf{p}$, the set $[n]$ of queues partitions into subsets $S_1,S_2,\ldots$, where each queue in $S_i$ group has the same aging rate and $S_1$ ages the fastest, then $S_2$, etc, according to $r$ (and so for $c$ as well).

To get a sense of the quantities that will arise before considering the general case, consider the simplest setting of a single queue and a single server (where there are no nontrivial strategies nor competition), with rates $\lambda>\mu$. In any round where the queue has an uncleared packet, the age will first increase by $1$ deterministically. With probability $\mu$, the queue will succeed in clearing this packet, and the age will go down in expectation by $\mathbb{E}[X]=1/\lambda$, where $X\sim \text{Geom}(\lambda)$ is independent of whether or not the server succeeds. Therefore, the expected change in this queue's age will be $1-\mu/\lambda>0$, and we expect that the queue will asymptotically age at this rate.

In general, with multiple queues and servers, the actual values of $c_i$ are best described via a recursive algorithm that computes the rates, which we give below. The intuition is that $S_1(\mathbf{p})$ will be the subset that minimizes 
the ratio of expected packets they clear collectively given $\mathbf{p}$, assuming they have priority over all other queues, divided by their sum of arrival rates. This quantity arises by viewing this subset as a single large queue as in the single queue 
example above. Conditioned on this set $S_1$ of queues growing fastest, they will typically have priority, and then we recurse to find the lower groups. The algorithm begins by initializing $k=1$ and $I=[n]$:

\begin{enumerate}
    \item  Compute the minimum value over all nonempty subsets $S\subseteq I$ of
    \begin{equation*}
        \frac{\sum_{j=1}^m \mu_j(1-\prod_{i\in S}(1-p_{i,j}))}{\sum_{i\in S} \lambda_i}.
    \end{equation*}
    This gives the \emph{expected number of packets cleared by $S$ if all queues in $S$ send in a time step and they have priority over all other queues, divided by their sum of arrival rates.}
    
    \item If this value is at least $1$, then no subset of queues will have linear aging, so set $S_k=I$,  
    $r_i(\mathbf{p})=0$ for all $i\in S_k$, and terminate. Otherwise, set $S_k$ to be the minimizer of the previous quantity over all nontrivial subsets of $I$, chosen to be of largest cardinality in the case of degeneracies.\footnote{We show in \Cref{lem:closure} that this choice is unique and canonical.} 
    In this case, for each $i\in S_k$, $r_i(\mathbf{p})$ gets set to 
    \begin{equation*}
        1-\frac{\sum_{j=1}^m \mu_j(1-\prod_{i\in S_k}(1-p_{i,j}))}{\sum_{i\in S_k} \lambda_i}.
    \end{equation*}
    For $k=1$, we refer to \emph{any} subset with the minimum ratio  as a \textbf{tight}, or \textbf{minimizing}, subset.
    
    \item Update the server rates $\mu_j$ as $
        \mu_j\leftarrow \mu_j\prod_{i\in S_k}(1-p_{i,j}).$
    That is, $\mu_j$ gets discounted by the probability a queue from $S_k$ sends to server $j$ (assuming all these queues are sending).
    Update $I\leftarrow I\setminus S_k$, $k\leftarrow k+1$, and recurse on $I$ with $\bm{\mu}$ and $\mathbf{p}_I$ if nonempty.
\end{enumerate}

As many of these quantities will appear often, we make the following conventions: for any subsets $S,S'$ such that $S\subseteq [n]\setminus S'$, 
define
$\lambda(S)$ as the sum of arrival rates of packets to a set of queues $S$, and $\alpha(S\vert \mathbf{p},\bm{\mu},S')$ as the expected number of packets cleared from queues in $S$ with service rates $\bm{\mu}$, if the queues in $S'$ have priority, $S$ has priority over all other queues, and all queues in $S\cup S'$ send packets in the round: 

\begin{equation*}
    \alpha(S\vert \mathbf{p},\bm{\mu},S')\triangleq \sum_{j=1}^m \mu_j\prod_{i\in S'}(1-p_{i,j})(1-\prod_{i\in S}(1-p_{i,j})), \quad\quad\quad
    \lambda(S) \triangleq \sum_{i\in S} \lambda_i,
\end{equation*}
and then let
\begin{equation*}
    f(S\vert \mathbf{p},\bm{\mu},\bm{\lambda},S')\triangleq \frac{\sum_{j=1}^m \mu_j\prod_{i\in S'}(1-p_{i,j})(1-\prod_{i\in S}(1-p_{i,j}))}{\sum_{i\in S}\lambda_i} = \frac{\alpha(S\vert \mathbf{p},\bm{\mu},S')}{\lambda(S)},
\end{equation*}
denote the ratio of expected number of packets cleared by $S$ when having priority over all members but $S'$, normalized by the expected number of new packets received in each round by $S$. 
Let $S_{k}(\mathbf{p},\bm{\mu},\bm{\lambda})$ be the ${k}$th set output by the above algorithm.
When $\mathbf{p},\bm{\mu},\bm{\lambda}$ are clear from context, we will suppress them. We write $U_k=\cup_{\ell=1}^kS_{\ell}$ as the set of queues in the top $k$ groups outputted by the algorithm, with $U_0=\emptyset$. We will write $f_k = f(S_k\vert U_{k-1})$, and we use $g_k=\max\{0,1-f_k\}$ for the rate of the $k$th outputted set, which is equal 
to  $r_i(\mathbf{p},\bm{\mu},\bm{\lambda})$ for any $i\in S_k(\mathbf{p},\bm{\mu},\bm{\lambda})$.
From the recursive construction,
\begin{equation}
\label{eq:recurrence}
S_{k+1}(\mathbf{p},\bm{\mu},\bm{\lambda})=S_1(\mathbf{p}_{[n]\setminus U_k},\bm{\mu}',\bm{\lambda}_{[n]\setminus U_k})
\end{equation}
where $
    \mu'_j = \mu_j\prod_{i\in U_k(\mathbf{p},\bm{\mu},\bm{\lambda})}(1-p_{ij})
$ for all $j\in [m]$.
 In words, having found $U_k$, $S_{k+1}$ is the largest minimal set among the remaining elements, but where the $\bm{\mu}$ rates have been reweighed by the probability no element of $U_k$ sends to each server. These quantities are compiled in a table in \Cref{sec:table} for easy reference.

\subsection{Properties of $r$}
\label{section:basicproperties}
We first record basic properties of the output of the algorithm that will be useful in studying the analytic properties, not to mention in proving that this algorithm gives the correct asymptotic rates. Clearly, for fixed $S$, the function $f(S\vert T)$ is nonincreasing in $T$ as a set function. We repeatedly use the following fact, which can be seen simply by expanding the definition of $f$:

\begin{fact}
\label{fact:fracdecomp}
Suppose $S,S',T$ are such that $S,S'\subseteq [n]\setminus T$ and are disjoint. Writing $f$ in the form of the quotient $\alpha/\lambda$, then
\begin{equation*}
    f(S\cup S'\vert T)=f(S\vert T)\oplus f(S'\vert S\cup T).
\end{equation*}
\end{fact}
\emph{Throughout this paper, we will view $f$ as the quotient $\alpha/\lambda$ when invoking \Cref{fact:fracdecomp}.}

Next, we characterize some structure in the minimizing subsets at each step of the algorithm, which will allow us to choose the $S_k$ canonically as the largest cardinality maximizer. To do this, we first show that the function $\alpha(\cdot)$ is \emph{submodular}:

\begin{lem}[Submodular]
\label{lem:submodular}
For fixed $S'$, $\mathbf{p},\bm{\mu},\bm{\lambda}$ the function $\alpha(S\vert \mathbf{p},\bm{\mu},\bm{\lambda},S')$ is submodular is $S$, i.e. for any $S,T\subseteq [n]\setminus S'$, $\alpha(S\cap T\vert S')+\alpha(S\cup T\vert S')\leq \alpha(S\vert S')+\alpha(T\vert S')$.
\end{lem}
\begin{proof}
Fix $S',\mathbf{p},\bm{\mu},\bm{\lambda}$, and we supress the dependence in $\alpha$. To prove submodularity, recall that an equivalent definition is that for any $S,T$ satisfying $S\subseteq T$, and $i\not\in T$, then $\alpha(S\cup \{i\})-\alpha(S)\geq \alpha(T\cup \{i\})-\alpha(T)$ \cite{schrijver2003combinatorial}. Let $S\subseteq T$ and $i\not\in T$. A simple computation shows that for any $V\subseteq [n]\setminus S'$,
\begin{equation*}
    \alpha(V\cup \{i\})-\alpha(V)=\sum_{j=1}^m \mu'_j p_{ij}\prod_{k\in V}(1-p_{kj}),
\end{equation*}
where $\mu'_j$ is $\mu_j \cdot \prod_{i\in S'}(1-p_{ij})$ from priority. As each factor in the product is at most $1$, this is clearly decreasing in $V$ as a set function, establishing submodularity.
\end{proof}

Now, recall that the relevant functions in the construction of the above algorithm is the set function $f=\alpha/\lambda$. As a consequence of the fact that this function is the ratio of a submodular function with a modular function, we will be able to gain significant closure properties of the tight subsets (as defined above), which will end up being critical in establishing both game-theoretic and probabilistic properties of our systems.\footnote{In an unrelated context, similar ideas were used by Benjamini, et al. to show that the edge-isoperimetric ratio is not attained in certain infinite networks \cite{benjamini1999group,MR3616205}.} 
\begin{lem}[Closure]
\label{lem:closure}
For each fixed $\mathbf{p}$ and $k\geq 1$, the set of minimizers of $f(\cdot \vert U_{k-1})$ in $\mathcal{P}([n]\setminus U_{k-1})$ are closed under union and non-disjoint intersection; that is, if $S,S'\subseteq [n]\setminus U_{k-1}$ are minimizers, then
so is $S\cup S'$, as well as $S\cap S'$ if nonempty. Moreover, if $S\cap S'$ is empty, then the queues in $S$ and $S'$ must send to disjoint subsets of servers.\footnote{While not necessary for our results, one can use the proof of \Cref{lem:closure} to show that the maximizing subsets can be computed in strongly polynomial time using the fact that the function $\alpha(\cdot)-y\cdot \lambda(\cdot)$ is submodular.}

In particular, the minimizing set with largest cardinality is unique, and is the union of all minimizing sets at step $k$. If $S$ is considered at step $k$ of the algorithm, but $S$ is not a subset of $S_k$, then $f(S\vert U_{k-1})> f(S_k\vert U_{k-1})$.
\end{lem}
\begin{proof}
The last
statement is an immediate consequence of the first, so we focus on the first part. We show this just for $k=1$; the general case for $k>1$ follows from 
the recurrence in (\ref{eq:recurrence}).

By \Cref{lem:submodular}, $\alpha$ is submodular, and it is immediate that $\lambda(\cdot)$ is a \emph{modular} function, i.e. $\lambda(S\cup S')+\lambda(S\cap S')=\lambda(S)+\lambda(S')$. We claim that if a function $f$ is a ratio of a nonnegative submodular function and a nonnegative supermodular function, then the set of minimizers of $f$ is closed under union and non-disjoint intersection. To see this, suppose $S,S'$ are minimizers. Then we see the following inequalities, writing $f$ always in the form of the quotient $\alpha/\lambda$, and using \Cref{fact:fracdecomp}  and \Cref{fact:dansineq} 
\begin{equation*}
     \max\{f(S),f(S')\} 
     \ge f(S)\oplus f(S') \geq f(S\cap S')\oplus f(S\cup S') \geq \min\{f(S\cap S'),f(S\cup S')\}
\end{equation*}
(where we omit $S\cap S'$ if it is empty), as the inequalities in the numerator and denominator go the correct way by super/submodularity, and we use \Cref{fact:dansineq}. But as $S,S'$ were minimizers, these must be equalities, which occurs iff $S\cap S'$ (if nonempty) and $S\cup S'$ are both minimizers. As $f=\alpha/\lambda$ here, this applies for our functions.

If $S,S'$ are both minimizers and are disjoint, then from \Cref{fact:fracdecomp}, it follows that $f(S\cup S') = f(S) \oplus f(S'\vert S)$. As $S,S',S\cup S'$ are evidently minimal, this equation can only hold if $f(S'\vert S)=f(S)$, which occurs if and only if $S$ and $S'$ disjointly mix among servers.
\end{proof}

From \Cref{lem:closure}, it will nearly immediately follow that the outputted rates are strictly monotonic decreasing in the groups:
as mentioned, $[n]=S_1\sqcup S_2\sqcup\ldots$ is meant to give a partition into groups that age together, where $S_1$ is the fastest aging group, $S_2$ the next fastest, etc. As such, the disjoint subsets iteratively output by the algorithm satisfy the intuition that motivates the construction.

\begin{lem}[Monotonicity]
\label{lem:monotonicity}
Let $S_1,S_2,\ldots$ be the outputs of the algorithm in order. Then $g_k> g_{k+1}$ for each $k\geq 1$.
\end{lem}
\begin{proof}
Consider $S_{k}\cup S_{k+1}$: this set is considered at stage $k$ of the algorithm, and evidently was not selected. From the previous lemma, we must have
\begin{equation*}
    f(S_{k}\vert U_{k-1})< f(S_{k}\cup S_{k+1}\vert U_{k-1})<  f(S_{k+1}\vert U_{k-1}\cup S_k)=f(S_{k+1}\vert U_{k}).
\end{equation*}
The first inequality follows from the selection criteria of the algorithm (and maximality of $S_{k}$), while the second follows from writing (via \Cref{fact:fracdecomp})
\begin{equation*}
    f(S_{k}\cup S_{k+1}\vert U_{k-1})=f(S_{k}\vert U_{k-1})\oplus f(S_{k+1}\vert U_{k-1}\cup S_{k}).
\end{equation*}
As we have already proven the left hand side of the inequality, the right hand side must follow from \Cref{fact:dansineq}. This implies the claim.
\end{proof}

With these two basic properties, we can obtain an important structural result that will prove fruitul in establishing the existence of equilibria in the next section: 

\begin{thm}[Continuity of $r$]
\label{thm:ratecontinuity}
The function $r:(\Delta^{m-1})^n\to [0,1]^n$ given by $r(\mathbf{p})=(r_1(\mathbf{p}),\ldots,r_n(\mathbf{p}))$ is continuous.
\end{thm}
\begin{proof}
Fix $\mathbf{p}^*$, a point we wish to show continuity at, and let $\mathbf{p}^k\to \mathbf{p}^*$ be a convergent sequence in $(\Delta^{m-1})^n$. It is easy to see that because the function $\mathbf{p}\mapsto \max_{S\subseteq [n]} f(S\vert \mathbf{p})$ is clearly continuous as the maximum of continuous functions, the function $\mathbf{p}\mapsto \max\{ 1-f(S_1(\mathbf{p})\vert \mathbf{p}),0\} = g(S_1(\mathbf{p}))$ is continuous. Therefore, if $\max_{i\in [n]} r_i(\mathbf{p}^*)=g_1(\mathbf{p}^*)=0$, then by \Cref{lem:monotonicity}, as $\max_i r_i(\mathbf{p}^k)\to 0$, monotonicity yields $r_i(\mathbf{p}^k)\to 0$ along the sequence for every $i\in [n]$, proving continuity. We will now assume the harder case $g_1(\mathbf{p}^*)>0$.

Before proceeding, we define $\delta>0$ to be the minimal nonzero gap between $f(S\vert S',\mathbf{p}^*)$ and $f(T\vert S',\mathbf{p}^*)$ over all choices of $S,T, S'$ such that $S,T\subseteq [n]\setminus S'$, i.e.
\begin{equation}
\label{eq:mingap}
    \delta \triangleq \min_{S,T,S': f(S\vert S',\mathbf{p}^*)\neq f(T\vert S',\mathbf{p}^*)} \vert f(S\vert S',\mathbf{p}^*)- f(T\vert S',\mathbf{p}^*) \vert. 
\end{equation}
Note that $\delta$ is strictly positive as there are only finitely many choices of $S,T,S'$.

 Now, fix $0<\varepsilon<\delta$, and recall that for any fixed $S,S'$ such that $S\subseteq [n]\setminus S'$, the function $f(S\vert \mathbf{p}, S')$ is continuous as a function of $\mathbf{p}$. In particular, for this choice of 
 $\epsilon$, we may restrict to a tail of the sequence $\{\mathbf{p}^k\}$ and reindex so that for all $k\geq 1$, and all $S,S'$
 \begin{equation}
 \label{eq:epsclose}
     \vert f(S\vert \mathbf{p}^k,S') - f(S\vert \mathbf{p}^*,S')\vert < \varepsilon/3.
 \end{equation}
 
 We now claim that for every $k\geq 1$, the following holds in the algorithm's outputted rates on $\mathbf{p}^k$: while there exists any element $i\in S_1(\mathbf{p}^*)$ that has not been outputted, the union of outputted subsets to that point must itself be a minimizing subset with respect to $f(\cdot)$ evaluated with profile $\mathbf{p}^*$, and that each element outputted so far has $f$ value at $p^k$ at most $f(S_1\vert \mathbf{p}^*)+\varepsilon/3$. 
 
 To see this, we proceed inductively: at the beginning of the algorithm, for every tight subset $S\subseteq S_1(\mathbf{p}^*)$ (by \Cref{lem:closure}), we have by \Cref{eq:epsclose}
 \begin{equation*}
     f(S\vert \mathbf{p}^k) <f(S\vert \mathbf{p}^*)+\varepsilon/3=f(S_1(\mathbf{p}^*)\vert \mathbf{p}^*)+\varepsilon/3,
 \end{equation*}
 while for every subset $T$ that is not minimal, we have by \Cref{eq:mingap,eq:epsclose} that
 \begin{equation*}
     f(T\vert \mathbf{p}^k) >f(T\vert \mathbf{p}^*)-\varepsilon/3\geq f(S_1(\mathbf{p}^*)\vert \mathbf{p}^*)+\delta-\varepsilon/3 >f(S_1(\mathbf{p}^*)\vert \mathbf{p}^*)+2\varepsilon/3.
 \end{equation*}
 In particular, the first outputted subset must be a tight subset for $\mathbf{p}^*$, and the rate of each element in that subset is at least the desired amount. 
 
 Suppose this holds inductively, and now let $S\subseteq S_1(\mathbf{p}^*)$ be the union of the initial outputted sets, which we know is tight. If $S=S_1(\mathbf{p}^*)$, we are done, so suppose there exists $i\in S_1(\mathbf{p}^*)\setminus S$. Suppose $T$ is disjoint and such that $S\cup T\subseteq S_1(\mathbf{p}^*)$ is tight. 
 Such sets exist, for instance $T=S_1(\mathbf{p}^*)\setminus S$.
 From \Cref{fact:fracdecomp},
 \begin{equation*}
     f(S\cup T\vert \mathbf{p^*}) = f(S\vert \mathbf{p}^*)\oplus f(T\vert S,\mathbf{p}^*),
 \end{equation*}
 and by \Cref{fact:dansineq}, the fact that both the left side and the first term on the right are minimal implies $f(T\vert S,\mathbf{p}^*)=f(S_1(\mathbf{p}^*)\vert \mathbf{p}^*)$.
 Then at the next step of the algorithm, we again have by \Cref{eq:epsclose}
 \begin{equation*}
     f(T\vert S,\mathbf{p}^k)< f(T\vert S, \mathbf{p}^*)+\varepsilon/3=f(S_1(\mathbf{p}^*)\vert \mathbf{p}^*)+\varepsilon/3.
 \end{equation*}
 For any disjoint subset $T'$ such that $S\cup T'$ is not tight, note that by \Cref{fact:fracdecomp},
 \begin{equation*}
     f(S\cup T'\vert \mathbf{p}^*) = f(S\vert \mathbf{p}^*)\oplus f(T'\vert S,\mathbf{p}^*).
 \end{equation*}
 By \Cref{fact:dansineq}, minimality of the first term on the right, and the fact that the left term is \emph{not} minimal, it follows that $f(T'\vert S,\mathbf{p}^*)\geq f(S_1(\mathbf{p}^*)\vert \mathbf{p}^*)+\delta$ by definition.
 We then have by \Cref{eq:epsclose} that
 \begin{equation*}
     f(T'\vert S,\mathbf{p}^k) >f(T'\vert S, \mathbf{p}^*)-\varepsilon/3\geq f(S_1(\mathbf{p}^*)\vert \mathbf{p}^*)+\delta-\varepsilon/3 >f(S_1(\mathbf{p}^*)\vert \mathbf{p}^*)+2\varepsilon/3.
 \end{equation*}
 In particular, the next outputted set is such that $S\cup T\subseteq S_1(\mathbf{p}^*)$ is tight, and the rate condition holds as well. We can iteratively apply this while $S_1(\mathbf{p}^*)$ is not exhausted, proving that every element of $S_1(\mathbf{p}^*)$ is outputted before any other element when running the algorithm on $\mathbf{p}^k$, with rate within $\varepsilon/3$ of $r_i(\mathbf{p}^*)$. As $\varepsilon$ was arbitrary, this shows continuity on all components of $S_1(\mathbf{p}^*)$.
 
 Because we have shown that every queue of $S_1(\mathbf{p}^*)$ is outputted before every queue not in $S_1$, we can apply the recurrence as discussed in \Cref{eq:recurrence} to show continuity for each queue in $S_2(\mathbf{p}^*)$, just discounting $\bm{\mu}$ as usual. The same argument restricted to $[n]\setminus S_1(\mathbf{p}^*)$ nearly shows continuity; the only difference is the discounting of $\bm{\mu}$ by the queues in $S_1(\mathbf{p}^*)$ depends on $\mathbf{p}^k$, not $\mathbf{p}^*$, but as each $f(S\vert S_1(\mathbf{p}^*),\mathbf{p},\bm{\mu}')$ is jointly continuous in $\mathbf{p},\bm{\mu}'$, and the composition of continuous functions is continuous, the same argument holds with minimal modification. This proves continuity for the components of each subsequent group recursively, and thus of each component in $[n]$.
\end{proof}

\subsection{Existence of Equilibria}
\label{section:equilibria}
With these structural results, we can turn to showing our first game-theoretic property of this game, for now assuming that the costs are given by $r$, the output of the algorithm of \Cref{sec:algdescription}: namely, that equilibria exist. While the cost functions are not quite convex, by restricting each component to a line that varies only a single queue's strategy, one can deduce enough structure that allows for an application of Kakutani's Theorem. We now carry out this plan.

We proceed as follows: fix any queue $i$, as well as any fixed probability choices $p_{-i}\in (\Delta^{m-1})^{n-1}$ by the other players, and any two $p,p'\in \Delta^{m-1}$. Define for $t\in [0,1]$
\begin{equation*}
    h(t)=r_i(tp+(1-t)p',p_{-i}).
\end{equation*}

\begin{lem}
\label{lem:nolocalmax}
For any fixed $i$, $p_{-i}\in (\Delta^{m-1})^{n-1}$, and $p_i,p'_i\in \Delta^{m-1}$, the function $h(t)$ is piecewise linear and has no local maxima on the interior.
\end{lem}
\begin{proof}
Let $\mathbf{p}(t)=(tp_i+(1-t)p'_i,p_{-i})$.
By \Cref{thm:ratecontinuity}, $h$ is continuous as the restriction of a continuous function, and is easy to see it must be piecewise linear in $t$ by inspection. Indeed, as the algorithmic description of $c$ takes minimums and maximums of finitely many linear functions, this yields a piecewise linear function with no jump discontinuities.

We now prove the last claim. It is sufficient to show that if $h$ is increasing at $t'$, then it is increasing for all $t''>t'$. Suppose that this is violated for some $t'<t''$; by piecewise linearity, there must exist some $t^*$ such that as $t'<t^*<t''$ where two lines intersect in the graph, and so that as $t\to t^{*-}$, the slope is increasing while it is non-increasing for $t\to t^{*+}$. 

Suppose that for all $t$ that are sufficiently close to $t^*$ from the left, $i$ is outputted at step $k$ of the algorithm. The only reason the slope can go from positive to nonpositive at $t^*$ is there is a change in which sets are outputted in the algorithm at some step $\ell\leq k$,  which can happen only if some new set $S$ including $i$ gets selected for $t\geq t^*$. However, as the rates of all sets not including $i$ fixing any other disjoint set having priority are all constants with respect to $t$, this can only occur because at $t^*$, some linear function $f(S\vert \mathbf{p}(t),S')$ went below the $f(S''\vert \mathbf{p}(t),S')$ that was previously selected at step $\ell$, where $S'$ is the union of all sets outputted prior at $t$ and $S''$ is the set that was outputted next for all $t$ close enough to the left of $t^*$. If $S''$ included $i$, this can only occur if $r(S\vert \mathbf{p}(t),S'):=\max\{0,1-f(S\vert\mathbf{p}(t),S')\}=1-f(S\vert \mathbf{p}(t),S')\}$ has larger positive slope than $r(S''\vert \mathbf{p}(t),S')$, so the slope of $h$ would be strictly larger (and in particular, increasing) for all $t$ sufficiently close to $t^*$ on the right, contradicting our assumption that it is non-increasing. If $S''$ does not include $i$, $r(S''\vert \mathbf{p}(t),S')$ is a constant with respect to $t$, so for $r(S\vert \mathbf{p}(t),S')$ to exceed it for $t$ larger than $t^*$ but be lower for $t$ less than $t^*$, the slope of $r(S\vert \mathbf{p}(t),S')$ must also be positive, another contradiction. Both cases lead to a contradiction, proving the claim.
\end{proof}

\begin{corollary}
\label{cor:globalmins}
Fix $p_{-i}\in (\Delta^{m-1})^{n-1}$. Then the set of global minimizers of $r_i(\cdot,p_{-i})$ 
form a nonempty, closed convex set.
\end{corollary}
\begin{proof}
Note that global minima exist by continuity from \Cref{thm:ratecontinuity} and the Extreme Value Theorem. Let $p_i,p'_i$ be global minimizers; if we form the line between them in $\Delta^{n-1}$ (which is of course convex) and consider the function $h$ \jgedit{defined on this line}, then as there are no local maxima in the interior of $h$ by the previous lemma, the maximum must lie at an endpoint. This immediately implies that every point on this line is also a global minimizer. \jgedit{Closedness of the set of global minimizers follows immediately from the continuity guaranteed in \Cref{thm:ratecontinuity}.}
\end{proof}

\begin{thm}[Existence of Nash Equilibria]
\label{thm:nashexists}
There exists a pure equilibrium of the game with costs given by $r:(\Delta^{m-1})^n\to [0,1]^n$.
\end{thm}
\begin{proof}
We will prove this by appealing to Kakutani's Theorem. 
Let $B:(\Delta^{m-1})^n\rightrightarrows (\Delta^{m-1})^n$ be the correspondence that maps $\mathbf{p} \in (\Delta^{m-1})^n$ to the set $B(\mathbf{p})\subseteq (\Delta^{m-1})^n$ where
\begin{equation*}
    B(\mathbf{p})=\{\mathbf{p}'\in (\Delta^{m-1})^n: p'_i\in \arg\min_{x\in \Delta^{m-1}} r_i(x,p_{-i})\}
\end{equation*}
is the best-response correspondence. 

We must verify the preconditions of Kakutani's Theorem. $(\Delta^{m-1})^n$ is clearly compact and convex, and we have shown that $B(\mathbf{p})$ is non-empty and is a convex set by \Cref{cor:globalmins}. The final condition to show is that $r$ has closed graph, which can be done by a completely standard argument; we must show that if $(\mathbf{p}^k,\mathbf{s}^k)\to (\mathbf{p},\mathbf{s})$, where $\mathbf{s}^k\in B(\mathbf{p}^k)$, then $\mathbf{s}\in B(\mathbf{p})$. Suppose for a contradiction that this does not hold for some such convergent sequence. This implies that for some $i\in [n]$, there exists some $s'_i$ and $\epsilon>0$ such that
\begin{equation*}
    r_i(s'_i,p_{-i})+3\epsilon< r_i(s_i,p_{-i}).
\end{equation*}
As $p_{-i}^k\to p_{-i}$, the continuity of $r$ from \Cref{thm:ratecontinuity} gives for large enough $k$ that
\begin{equation*}
    r_i(s'_i,p^k_{-i})\leq r_i(s'_i,p_{-i})+\epsilon.
\end{equation*}
These two inequalities yield
\begin{equation*}
    r_i(s'_i,p^k_{-i})<r_i(s_i,p_{-i})-2\epsilon<r_i(s^k_i,p^k_{-i})-\epsilon,
\end{equation*}
where the last inequality holds for all large enough $k$ by continuity of $r$. This contradicts the optimality of $\mathbf{s}^k\in B(\mathbf{p}^k)$, proving $r$ has closed graph.

Kakutani's Theorem then immediately implies the existence of a pure equilibrium, i.e. $\mathbf{p}\in (\Delta^{m-1})^n$ such that $\mathbf{p}\in B(\mathbf{p})$.
\end{proof}

\section{Asymptotic Convergence to $r$}
\label{section:convergence}
Having established game-theoretic properties of the patient queuing game assuming that the costs are given by the function $r$, rather than the random asymptotic limiting behavior $c$, we now return to the task of showing that these quantities are almost surely equal. Our main technical result asserts precisely this:

\begin{thm}[Almost Sure Asymptotic Convergence]
\label{thm:gamemain}
Let $\mathcal{G}=([n],(c_i)_{i=1}^n,\bm{\mu},\bm{\lambda})$ be a one-shot queuing game. For each fixed $\mathbf{p}$ and all $i\in [n]$, almost surely the long-run aging rate of queue $i$ satisfies
\begin{equation*}
    c_i(\mathbf{p})=\lim_{t\to \infty} \frac{T^i_t}{t}=r_i(\mathbf{p}).
\end{equation*}
\end{thm}

Before providing the proof, we give an overview of the details: the high-level idea is to show that this identity holds for all $i\in S_1$, then $S_2$, and so on. The first step is showing that the maximum queue age grows by at most the desired rate on each long-enough window with high probability (\Cref{prop:upperbound}); this can be done using the Azuma-Hoeffding inequality and other general concentration bounds. The key insight is that if a subset of much older queues $S$ is likely to have priority on a long window of length $w$, the quantity $w\cdot f(S_1)\cdot \lambda(S)$ is a \emph{lower bound} on the expected number of packets cleared \emph{collectively} by $S$ on this window by definition of $S_1$. In the case that there is a unique very old queue, this will be enough to establish an upper bound on the aging rate because each packet cleared decreases the age by $1/\lambda$ in expectation. The analysis gets more complicated when there are multiple old queues, as while we know these queues \emph{collectively} have priority over all young queues, we must argue about priorities \emph{within} this subset to bound the growth of the maximum queue age. We deal with this by induction by carefully chaining together large windows to obtain a win-win analysis.

Once we have established this upper bound on aging on all queues, we will use Azuma-Hoeffding to argue the \emph{average} queue in $S_1$ ages at a rate of at least $r(S_1)$ almost surely (\Cref{prop:lowerbound}). Combined with the upper bound, this will allow us to conclude that, because the average queue and oldest queue in $S_1$ ages at the desired rate almost surely, \emph{all} queues in $S_1$ must age at this rate almost surely. To extend this analysis to lower groups $S_2$, etc, we will use a similar analysis to show that the maximum age over every queue not in $S_1$ grows at most like $r(S_2)$. Then, because we know that every queue in $S_1$ grows by a $r(S_1)>r(S_2)$ rate, almost surely, eventually every queue in $S_1$ will be much older that every queue not in $S_1$, giving priority. We leverage this fact to show again that the average queue in $S_2$ must grow by at least $r(S_2)$, and therefore every queue in $S_2$ grows at this rate almost surely. The proof for the lower groups $S_3,\ldots$ is completely analogous.

We now carry out this high-level plan, though deferring the more technical intermediate results to \Cref{sec:appendixconv}. Our main intermediate claim asserts that with high probability, the maximum queue age increases at a rate of at most $(1-(1-\epsilon)\cdot f_1)$ on the next $w$ steps for a large enough $w$. In fact, more generally, the following holds:
\begin{proposition}
\label{prop:upperbound}
Fix $\epsilon>0$. For any integer $a\in \mathbb{N}$, let $w=a\cdot \lceil \frac{6}{\epsilon}\rceil^{n-1}$. Suppose it holds at time $t$ that $\max_{i\in [n]} T^i_t\geq w\cdot f_1$. Then 
\begin{equation*}
    \max_{i\in [n]} T^i_{t+w} -\max_{i\in [n]} T^i_t \leq (1-(1-\epsilon)\cdot f_1)\cdot w
\end{equation*}
with probability at least $1-C_1\exp(-C_2a)$, where $C_1,C_2>0$ are absolute constants depending only on $n,\epsilon,\bm{\lambda}, \bm{\mu},\mathbf{p}$, but not on $a$.

More generally, for each $s\geq 1$, if $\max_{i\not\in U_{s-1}} T^i_t\geq w\cdot f_s$, then 
\begin{equation*}
    \max_{i\not\in U_{s-1}} T^i_{t+w} -\max_{i\not\in U_{s-1}} T^i_t \leq (1-(1-\epsilon)\cdot f_s)\cdot w
\end{equation*}
with probability at least $1-C_1\exp(-C_2a)$, where $C_1,C_2>0$ are absolute constants depending only on $n,\epsilon,\bm{\lambda}, \bm{\mu},\mathbf{p}$, but not on $a$.
\end{proposition}

For \Cref{prop:upperbound} to yield anything useful, we will need a corresponding lower bound that asserts roughly that if groups have separated according to what the algorithm asserts, then the aging rate of the \emph{average} queue in a group grows at the conjectured rate. To that end, we show the following result which shows that, \emph{if} we have the conjectured separation between groups $U_{k-1}$ and $S_k$ (i.e. each queue in the former is significantly older than each queue in the latter), then some weighted combination of the queue ages in $S_k$ (whose significance will prove apparent momentarily) must rise significantly.

\begin{proposition}
\label{prop:lowerbound}
For any $s\geq 1$ and any fixed $\epsilon>0$, the following holds: suppose that at time $t$, it holds that
\begin{equation*}
    \min_{i\in U_s} T^i_t -\max_{i\in S_{s+1}} T^i_t\geq  2\cdot \frac{w}{\lambda_n}.
\end{equation*}
Then with probability $1-A\exp(-Bw)$ where $A,B>0$ are absolute constants not depending on $w$, we have
\begin{equation*}
    \sum_{i\in S_{s+1}} \lambda_i T^i_{t+w}-\sum_{i\in S_{s+1}} \lambda_iT^i_t\geq (1-(1+\epsilon)f_{s+1})\cdot w \cdot \left(\sum_{i\in S_{s+1}} \lambda_i\right).
\end{equation*}

Moreover, for any fixed $\epsilon>0$, with probability at least $1-A\exp(-Bw)$ it holds that
\begin{equation*}
    \sum_{i\in S_1} \lambda_iT^i_w \geq (1-(1+\epsilon)f_1)\cdot w\cdot \left(\sum_{i\in S_1} \lambda_i\right).
\end{equation*}
\end{proposition}

With these two results in hand, we may return to the proof of \Cref{thm:gamemain}.

\begin{proof}[Proof of \Cref{thm:gamemain}]
We will show the desired statement holds for each $i\in S_1$, then $S_2$, and so on. We first treat the case that the last outputted group $S_k$ satisfies $g_k=0$, or equivalently that $f_k\geq 1$.

Fix $\epsilon>0$ and partition time into consecutive windows of size $w_{\ell}=\ell\cdot \lceil\frac{6}{\epsilon}\rceil^{n-1}$. Let $W_{\ell}=\sum_{q=1}^{\ell-1} w_{q}$ be the time period at the beginning of the $\ell$th window, and note that $w_{\ell}=\Theta(W_{\ell}^{1/2})$.

Consider the following events for $\ell=1,2,\ldots$
\begin{gather*}
    A_{\ell}=\left\{\max_{i\in [n]\setminus U_{k-1}} T^i_{W_{\ell+1}} -\max_{i\in [n]\setminus U_{k-1}} T^i_{W_{\ell}} > (1-(1-\epsilon)\cdot f_{k})\cdot w_{\ell}\right\}\\
    B_{\ell}=\left\{\max_{i\in [n]\setminus U_{k-1}} T^i_{W_{\ell}}\geq w_{\ell}\cdot f_{k}\right\}\\
    C_{\ell} = A_{\ell}\cap B_{\ell}.
\end{gather*}
Clearly, $\Pr(C_{\ell})\leq \Pr(A_{\ell}\vert B_{\ell})$. But by \Cref{prop:upperbound}, we know that for some constants $C_1,C_2>0$ independent of $\ell$, that 
\begin{equation*}
    \Pr(A_{\ell}\vert B_{\ell})\leq C_1\exp(-C_2\cdot \ell).
\end{equation*}
Therefore, we have that
\begin{equation*}
    \sum_{\ell=1}^{\infty} \Pr(C_{\ell})\leq \sum_{\ell=1}^{\infty} C_1\exp(-C_2\cdot \ell)<\infty.
\end{equation*}
The first Borel-Cantelli lemma (\Cref{lem:borel}) thus implies that almost surely at most finitely many of the $C_{\ell}$ occur, or equivalently, almost surely for all but finitely many of the $\ell$, either
\begin{equation*}
    \max_{i\in [n]\setminus U_{k-1}} T^i_{W_{\ell+1}} -\max_{i\in [n]\setminus U_{k-1}} T^i_{W_{\ell}} \leq  (1-(1-\epsilon)\cdot f_{k})\cdot w_{\ell}
\end{equation*}
or 
\begin{equation*}
    \max_{i\in [n]\setminus U_{k-1}} T^i_{W_{\ell}}< w_{\ell}\cdot f_{k}.
\end{equation*}
 Observe that for each of the intervals where the latter holds, the value during the interval is at most $w_{\ell}\cdot f_{k} + w_{\ell+1}=O(W_{\ell}^{1/2})$. In particular, it is not difficult to see that almost surely $\max_{i\in [n]\setminus U_{k-1}} T^i_{W_{\ell}}$ is either $o(W_{\ell})$, in which case we are done, or grows by at most a rate of $(1-(1-\epsilon)\cdot f_{k})\cdot w_{\ell}$. Either way, as $\epsilon>0$ was arbitrary, we may take $\epsilon\to 0$ to deduce the desired result that almost surely\footnote{Technically, for each $\epsilon>0$, this implies it only for $t$ of the form $t=W_{\ell}$. But when $W_{\ell}\leq t< W_{\ell+1}$, $T^i_t$ cannot be more than $w_{\ell}$ from the values at $W_{\ell}$, so the difference in the numerator is $O(w_{\ell})=O(t^{1/2})=o(t)$ and thus vanishes in the limit.}
 
 \begin{equation*}
     \limsup_{t\to \infty}\frac{\max_{i\in [n]\setminus U_{k-1}} T^i_{t}}{t}=0=g_{k},
 \end{equation*}
 using $f_{k}\geq 1$. As ages of queues are  nonnegative, the lower bound of $0$ is trivial.

Now we turn to the rest of the groups, and we now assume that $g_{k}>0$
. We do this inductively. For $S_1$, fix $\epsilon>0$ and define
\begin{equation*}
    A_t=\left\{\sum_{i\in S_1} \lambda_iT^i_t < (1-(1+\epsilon)f_1)\cdot t\cdot \left(\sum_{i\in S_1} \lambda_i\right)\right\}.
\end{equation*}
By \Cref{prop:lowerbound}, we know $\Pr(A_t)\leq A\exp(-Bt)$ for some constants $A,B>0$ independent of $t$. Therefore,
\begin{equation*}
    \sum_{t=1}^{\infty} \Pr(A_t)<\infty,
\end{equation*}
from which the Borel-Cantelli lemma implies that almost surely, for all but finitely many $t$,
\begin{equation*}
    \sum_{i\in S_1} \lambda_iT^i_t\geq (1-(1+\epsilon)f_1)\cdot t\cdot \left(\sum_{i\in S_1} \lambda_i\right).
\end{equation*}
Taking $\epsilon\to 0$, we obtain almost surely
\begin{equation}
\label{eq:averagelb}
    \liminf_{t\to \infty} \frac{\sum_{i\in S_1} \lambda_iT^i_t}{t}\geq g_1\cdot \left(\sum_{i\in S_1} \lambda_i\right).
\end{equation}
Next, note that deterministically, we have from \Cref{fact:dansineq}
\begin{equation}
\label{eq:fracminmax}
    \min_{i\in S_1} T^i_t\leq \frac{\sum_{i\in S_1} \lambda_iT^i_t}{\sum_{i\in S_1} \lambda_i}\leq \max_{i\in S_1} T^i_t.
\end{equation}
In particular, we deduce that almost surely,
\begin{equation}
\label{eq:S1lb}
    \liminf_{t\to \infty} \frac{\max_{i\in S_1} T^i_t}{t}\geq g_1>0.
\end{equation}
For the upper bound, let the $w_{\ell}$ and $W_{\ell}$ be as before, and now define
\begin{gather*}
    A_{\ell}=\left\{\max_{i\in [n]} T^i_{W_{\ell+1}} -\max_{i\in [n]} T^i_{W_{\ell}} > (1-(1-\epsilon)\cdot f_1)\cdot w_{\ell}\right\}\\
    B_{\ell}=\left\{\max_{i\in [n]} T^i_{W_{\ell}}\geq w_{\ell}\cdot f_1\right\}\\
    C_{\ell} = A_{\ell}\cap B_{\ell}.
\end{gather*}
Again, $\Pr(C_{\ell})\leq \Pr(A_{\ell}\vert B_{\ell})$. By \Cref{prop:upperbound}, a now routine application of the Borel-Cantelli lemma implies that almost surely, for all but finitely many $\ell$, either
\begin{equation*}
    \max_{i\in [n]} T^i_{W_{\ell+1}} -\max_{i\in [n]} T^i_{W_{\ell}} < (1-(1-\epsilon)\cdot f_1)\cdot w_{\ell}
\end{equation*}
or
\begin{equation*}
    \max_{i\in [n]} T^i_{W_{\ell}}< w_{\ell}\cdot f_1.
\end{equation*}
But the latter event cannot happen infinitely often with positive probability, as this would imply $\max_{i\in [n]} T^i_{W_{\ell}}=o(W_{\ell})$ infinitely often with nonzero probability, which violates (\ref{eq:S1lb}). Therefore, it must be the case that almost surely, for all but finitely many $\ell$, the former event holds. This implies that almost surely
\begin{equation*}
    \limsup_{t\to \infty} \frac{\max_{i\in [n]} T^i_{t}}{t}\leq (1-(1-\epsilon)\cdot f_1);
\end{equation*}
taking $\epsilon\to 0$ implies that 
\begin{equation*}
    \limsup_{t\to \infty} \frac{\max_{i\in [n]} T^i_{t}}{t}\leq g_1.
\end{equation*}
As clearly the left side is an upper bound for the $\limsup$ of only those queues in $S_1$, almost surely
\begin{equation*}
    \limsup_{t\to \infty} \frac{\max_{i\in S_1} T^i_{t}}{t}\leq g_1.
\end{equation*}
Combining this with (\ref{eq:S1lb}), we finally deduce that almost surely
\begin{equation*}
    \lim_{t\to \infty} \frac{\max_{i\in S_1} T^i_{t}}{t}= g_1.
\end{equation*}
Finally, using (\ref{eq:averagelb}) and (\ref{eq:fracminmax}), we can also conclude that almost surely
\begin{equation*}
    \lim_{t\to \infty} \frac{\min_{i\in S_1} T^i_{t}}{t}= g_1.
\end{equation*}
Recall that $r_i(\mathbf{p})=g_1$ for all $i\in S_1(\mathbf{p})$ by definition of $g_1$. This proves the theorem for all queues in $S_1$.

We now show how to extend this inductively to higher values of $k$ with $g_{k}>0$. Suppose that we have shown for all $i\in U_{k-1}$ that the desired almost sure limit holds, and now consider $S_{k}$. A completely analogous argument using the windows $w_{\ell}$ as above with \Cref{prop:upperbound} via the Borel-Cantelli lemma implies that almost surely
\begin{equation}
\label{eq:Sjub}
    \limsup_{t\to \infty} \frac{\max_{i\in [n]\setminus U_{k-1}}} T^i_t{t}\leq g_{k}.
\end{equation}
Now, with these same windows, fix $\epsilon>0$ and let
\begin{gather*}
    A_{\ell} = \left\{\sum_{i\in S_{k}} \lambda_i T^i_{W_{\ell+1}}-\sum_{i\in S_{k}} \lambda_iT^i_{W_{\ell}}< (1-(1+\epsilon)f_{k})\cdot w_{\ell} \cdot \left(\sum_{i\in S_{k}} \lambda_i\right)\right\}\\
    B_{\ell} = \left\{\min_{i\in U_{k-1}} T^i_{W_{\ell}} -\max_{i\in S_{k}} T^i_t\geq  2\cdot \frac{w_{\ell}}{\lambda_n}\right\}\\
    C_{\ell} = A_{\ell}\cap B_{\ell}.
\end{gather*}
Another completely analogous application of \Cref{prop:lowerbound} and the Borel-Cantelli lemma implies that almost surely, at most finitely many of the $C_{\ell}$ occur. That is, almost surely, for all but finitely many $\ell$, either 
\begin{equation*}
    \sum_{i\in S_{k}} \lambda_i T^i_{W_{\ell+1}}-\sum_{i\in S_{k}} \lambda_iT^i_{W_{\ell}}\geq (1-(1+\epsilon)f_{k})\cdot w_{\ell} \cdot \left(\sum_{i\in S_{k}} \lambda_i\right)
\end{equation*}
or 
\begin{equation*}
    \min_{i\in U_{k-1}} T^i_{W_{\ell}} -\max_{i\in S_{k}} T^i_t<  2\cdot \frac{w_{\ell}}{\lambda_n}.
\end{equation*}
But the latter event cannot happen infinitely often with any nonzero probability by virtue of the inductive hypothesis and (\ref{eq:Sjub}), as $g_{k-1}>g_{k}$ by \Cref{lem:monotonicity}, which implies that these timestamps cannot be so close infinitely often. Therefore, it must be the case that for all but finitely many of the $\ell$, the former event holds. As usual, this immediately implies that 
\begin{equation*}
    \liminf_{t\to \infty}\frac{\sum_{i\in S_{k}} \lambda_i T^i_{t}}{t}\geq (1-(1+\epsilon)f_{k})\left(\sum_{i\in S_{k}} \lambda_i\right).
\end{equation*}
Again taking $\epsilon\to 0$ thus implies that almost surely
\begin{equation*}
    \liminf_{t\to \infty}\frac{\sum_{i\in S_{k}} \lambda_i T^i_{t} }{t}\geq  g_{k}\left(\sum_{i\in S_{k}} \lambda_i\right),
\end{equation*}
which again coupled with (\ref{eq:Sjub}) and \Cref{fact:dansineq} yields that almost surely
\begin{equation*}
    \lim_{t\to \infty} \frac{\min_{i\in S_{k}} T^i_{t}}{t}=\lim_{t\to \infty} \frac{\max_{i\in S_{k}} T^i_{t}}{t}=  g_{k}.
\end{equation*}
The extension to all $i\in S_{k}$ follows in the same manner as before by comparing with the average.

This completes the proof of the theorem.
\end{proof}

Observe that \Cref{thm:gamemain} rather strongly and explicitly characterizes the \emph{linear} almost sure asymptotic growth rates of each queue for any choices of randomizations. Our main result in \Cref{thm:ebound1} will show that, with a small slack in the system capacity, each queue will be guaranteed \emph{sublinear} asymptotic growth almost surely in any equilibrium. While our objective function emphasizes the physical interpretation for each queue as asymptotic linear growth rate, these incentives impose that queues are indifferent between sublinear growth rates. One could instead define the game just using the $f_k$ quantities directly, rather than taking the max with $0$ as is needed to argue about the asymptotic growth rates via $r$. If queues started out equally backed up, the $f_k$ quantities measure the linear speed at which their ages descend to zero. In this setting, we provide the following stronger conclusion, whose proof is deferred to \Cref{sec:appendixconv}:

\begin{thm}
\label{thm:gamess}
Fix $\mathbf{p}$ and suppose that for some group $S_k$ output by the algorithm, $f_k>1$, so that $1-f_k<0$. Then, for each $i\in S_k$, $T^i_t$ is strongly stable.
\end{thm}

\section{Price of Anarchy}
\label{section:poa}
Having established the almost sure asymptotic convergence of this system for any fixed strategies and the existence of pure equilibria, we finally turn to the game-theoretic problem of understanding what condition ensures the stability at any equilibrium profile. By  considering
deviations by a queue at the Nash equilibrium to a single other server, one can show that the price of anarchy is always at most $2$, matching the learning bound of \cite{DBLP:conf/sigecom/GaitondeT20}. We show that this factor is loose with patience, and that $\frac{e}{e-1}\approx 1.58$ is the right factor by considering continuous deviations. 

The following simple example shows that this is the best possible constant factor: fix $\epsilon>0$ small and suppose there are $n$ queues and $n$ servers, with $\bm{\lambda}=(1-1/e+\epsilon,\ldots, 1-1/e+\epsilon)$ and $\bm{\mu}=(1,\ldots,1)$, and $\mathbf{p}$ has every queue uniformly mixing among the servers. It is easy to see by symmetry that this system is Nash with $S_1=[n]$, for if a queue deviates from this uniform distribution, this does not change the worst ratio. 
Moreover, for any fixed $\epsilon>0$, as $n\to \infty$, this system becomes unstable. One can check that 
\begin{equation*}
    f(S_1\vert \mathbf{p})= f([n]\vert \mathbf{p})=\frac{\sum_{j=1}^n (1-\prod_{i=1}^n (1-1/n))}{n(1-1/e+\epsilon)}\to \frac{1-1/e}{1-1/e+\epsilon}<1,
\end{equation*}
so that $r(S_1)=\max_i c_i(\mathbf{p})>0$. Our main result asserts that this is the worst case, where every queue is maximally colliding subject to being Nash. Concretely, we prove the following instance-dependent bound from which the claimed factor immediately follows:

\begin{thm}[Main]
\label{thm:ebound1}
Let $\mathbf{p}$ be any Nash equilibrium of $\mathcal{G}$, and let $S_1$ be as defined before. Then\footnote{For any fixed $k$, it is not difficult to determine the optimal value of $\mathbf{x}$ to give the tightest lower bound. It suffices to maximize the numerator, which is concave. By standard KKT conditions at optimality, for all $j,j'\in [m]$ such that $x_j>0$, we must have $
    \mu_j (1-x_j/k)^{k-1} = \mu_{j'} (1-x_{j'}/k)^{k-1}$,
and $x_{\ell}=0$ for all lower indices.}
\begin{equation*}
    f(S_1\vert \mathbf{p}) \geq \min\left\{1\,,\,\min_{k\leq n} \max_{\mathbf{x}\in \mathbb{R}^m_{\geq 0}: \sum_{j=1}^m x_j = k} \frac{\sum_{j=1}^m\mu_j (1-(1-x_j/k)^k)}{\sum_{i=1}^k \lambda_i}\right\}.
\end{equation*}
\end{thm}

\begin{corollary}
\label{thm:ebound}
Let $\mathbf{p}$ be a Nash equilibrium of $\mathcal{G}$, and suppose that for each $1\leq k\leq n$,
\begin{equation*}
    \sum_{j=1}^k \mu_j > \left(\frac{e}{e-1}\right)\sum_{i=1}^k \lambda_i.
\end{equation*}
Then every queue is stable at $\mathbf{p}$. In particular, the price of anarchy of the patient queuing game is exactly $\frac{e}{e-1}$.
\end{corollary}
\begin{proof}
In \Cref{thm:ebound1}, for any $k\leq n$, we may set $x_j=1$ for $1\leq j\leq k$. Note that $(1-x_j/k)^k< e^{-1}$, so if $\bm{\mu}$ majorizes $\bm{\lambda}$ by a factor of at least $\frac{e}{e-1}$, \Cref{thm:ebound1} implies that $f(S_1\vert \mathbf{p})\geq 1$. From \Cref{lem:monotonicity} and \Cref{thm:gamemain}, we conclude that all queues are stable.
\end{proof}

We now prepare for the proof of \Cref{thm:ebound1}. The idea will be to continuously deform the Nash profile towards a highly symmetrized strategy vector while only weakly decreasing $f(S_1)$. At the end of this process, we obtain a lower bound on this value at Nash. To do this deformation and ensure monotonicity of the growth rate, we must at some point use the Nash property. The difficulty in proving the tightness of this example lies in the form of the $f$ functions; recall that as $S_1$ is the set of all queues growing at the fastest rate as the union of all tight subsets, it can have many proper tight subsets, and each queue $i\in S_1$ thus has to locally optimize all of the functions $f(S\vert \mathbf{p})$ with $S\ni i$ \emph{simultaneously} at Nash (see \Cref{fig:fig1} for an interesting example). In particular, whenever a queue $i\in S_1$ is at Nash, one possible deviation may weakly decrease $f(S)$ for some tight subset $S\ni i$, while another deviation may be unprofitable because $f(S')$ weakly decreases for some \emph{different} tight subset $S'\ni i$. That is, each queue may be constrained by multiple different objective functions at Nash, making it difficult to generically argue about \emph{why} any given deviation decreases performance. We overcome this barrier by connecting the incentives for each queue in $S_1$ with the structure guaranteed by \Cref{lem:closure}. More concretely, we reduce the number of constraints we must consider for each queue by showing in \Cref{prop:levels} that there exists a much smaller, ordered set of tight subsets that completely characterizes the incentives of any queue: 

\begin{figure}[h]
\begin{center}
\includegraphics[trim = {0 2.4cm 0 0},scale = .3,clip]{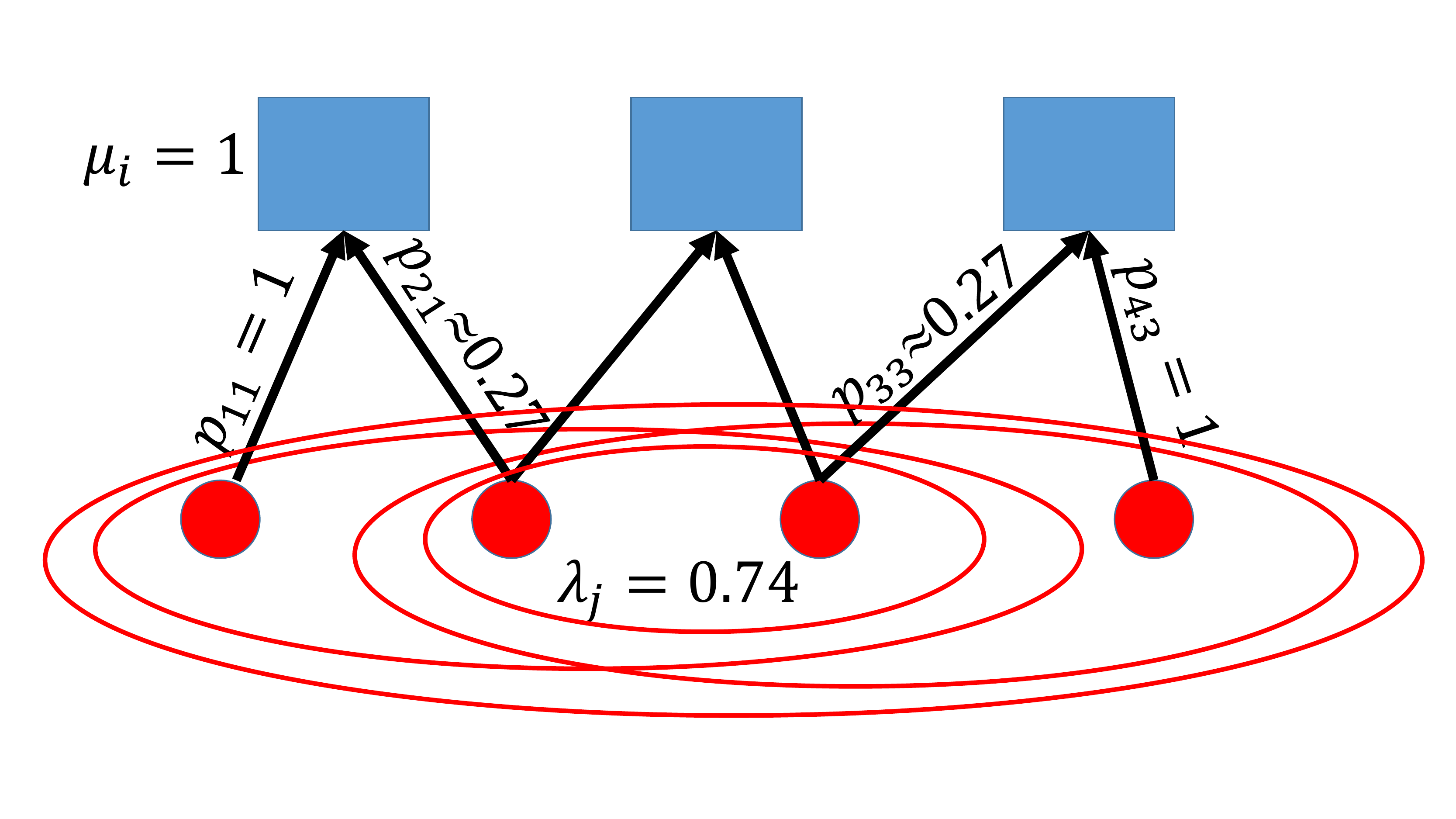}
\end{center}
\caption{\footnotesize Four queues compete for three servers, with maximally tight sets marked.
The two outer queues exclusively send to the outer servers, and the inner queues send to two servers each, as indicated by the figure. One can verify that this profile is Nash and the subsets marked are tight. Consider whether one of the inner queues would want to deviate from the current profile. If she shifts probability to the inner server, the smallest tight set will rise in aging rate. However, if she instead attempts to deviate to one of the outer servers, the rate of all four queues will rise.}
\label{fig:fig1}
\end{figure}

\begin{proposition}
\label{prop:levels}
Let $\mathbf{p}$ be any arbitrary strategy vector by the queues, and without loss of generality, let $[k]$ be the maximal tight subset after relabeling. Then, for some $s>0$, there exists a \textbf{level partition} of $[k]$ into $s$ levels with the following property: if a queue 
$i\in [k]$ belongs to a level-$\ell$ subset, then 
for any deviation by $i$ that shifts probability mass from one server to another and does not increase $f(S)$ for some tight subset $S\ni i$, there exists a tight subset $S'$ containing all queues at all levels $j\leq \ell$ such that $f(S')$ must not increase. 
\end{proposition}

\begin{proof}
We construct the desired decomposition of $[k]$ by instead recursively considering \emph{minimal} tight subsets. In particular, let the level-1 subsets be the set of all minimal tight subsets (notice that by \Cref{lem:closure} and minimality, each subset of queues must be disjoint and disjointly mixing). Note that to clear packets at the rate of these subsets, they all must have priority over any other queues.
Then given the level-$j$ subsets for all $j\leq \ell$ for some $\ell\geq 1$, we recursively define the level-$(\ell+1)$ subsets as the set of minimal tight subsets of the remaining queues conditioned on all subsets at lower levels having priority. The same argument as used by \Cref{lem:closure} implies that these subsets at each level must all be disjoint and disjointly mixing.\footnote{In the example given before of a system with multiple tight subsets, the level-1 subset is the subset of two queues that split between the inner and outer servers. The level-2 subsets are the two singleton sets of the outer queues sending each to their own outer server. Notice that these two subsets indeed disjointly mix.} It is easy to verify by iteratively using \Cref{fact:fracdecomp} that for any level-$\ell$, the union of all subsets at levels $j<\ell$ with any subset of the level-$\ell$ subsets must be tight. As $[k]$ is tight and any tight subset must be contained in $[k]$ by the maximality guaranteed in \Cref{lem:closure}, it is immediate that this decomposition exhausts $[k]$. 

Let $i\in [k]$ belong to some level-$\ell$ subset, for some $\ell\geq 1$. Suppose that $i$ has some deviation from server $j$ to another server $j'$ that causes $f(S)$ to not increase for some tight subset $S\subseteq [k]$. As $f(S)$ is a linear function in the randomizations by queue $i$ holding each other queue fixed, this must hold for any sufficiently small deviation as well. By expanding the partial derivatives $\frac{\partial f(S)}{\partial p_{ij}}$ and $\frac{\partial f(S)}{\partial p_{ij'}}$ and clearing denominators, this holds if and only if
\begin{equation}
\label{eq:Sminimal}
    \mu_j\prod_{r\in S\setminus \{i\}}(1-p_{rj})\geq \mu_{j'}\prod_{r\in S\setminus\{i\}}(1-p_{rj'})
\end{equation}
Define $V^i_{\leq \ell}$ to be the set of all queues at level at most $\ell$ that do not belong to the level-$\ell$ subset $i$ belongs to. 
Then $V^i_{\leq \ell}$ is tight, and by construction $S\setminus V^i_{\leq \ell}$ is nonempty as $i$ belongs to this set. As tight subsets are preserved under unions and intersections, we have $S\cup V^i_{\leq \ell}$ and $S\cap V^i_{\leq \ell}$ are tight. On the one hand, we have by \Cref{fact:fracdecomp} that
\begin{equation*}
    f(S) = f(S\cap V^i_{\leq \ell})\oplus f(S\setminus V^i_{\leq \ell}\vert S\cap V^i_{\leq \ell}).
\end{equation*}
On the other, we have
\begin{equation*}
    f(S\cup V^i_{\leq \ell}) = f(V^i_{\leq \ell}) \oplus f(S\setminus V^i_{\leq \ell}\vert V^i_{\leq \ell}).
\end{equation*}
Comparing these two expressions, the left-hand sides are minimal, as well as the first terms on the right, which by \Cref{fact:dansineq} implies the latter terms are equal. But we know from definition that
\begin{equation*}
    f(S\setminus V^i_{\leq \ell}\vert S\cap V^i_{\leq \ell})\geq f(S\setminus V^i_{\leq \ell}\vert V_{\leq \ell}),
\end{equation*}
with equality if and only if $V^i_{\leq \ell}\setminus S$ disjointly mixes from $S\setminus V^i_{\leq \ell}$ in $\mathbf{p}$.  As $i$ belongs to this set, we can combine this with \Cref{eq:Sminimal} to see
\begin{equation*}
    \mu_j\prod_{r\in (V^i_{\leq \ell}\cup S)\setminus \{i\}}(1-p_{rj})= \mu_j\prod_{r\in S\setminus \{i\}}(1-p_{rj})\geq\mu_{j'}\prod_{r\in S\setminus\{i\}}(1-p_{rj'})\geq \mu_{j'}\prod_{r\in (V^i_{\leq \ell}\cup S)\setminus\{i\}}(1-p_{rj'}).
\end{equation*}

This implies that $f(S\cup V^i_{\leq \ell})$ must also not increase from this deviation. The last thing to check is that the level-$\ell$ subset $i$ belongs to is contained in $S$, so that $S\cup V^i_{\leq \ell}$ contains all queues up to level $\ell$. But this follows because $S\setminus V^i_{\leq \ell}$ consists of queues at level at least $\ell$ and intersects this level-$\ell$ subset nontrivially at $i$, and therefore must contain it by minimality in our construction. Therefore, for this choice of deviation by queue $i$, the subset $S\cup V^i_{\leq \ell}$ certifies the desired claim.
\end{proof}

With this result, we may finally return to the proof of \Cref{thm:ebound1}.

\begin{proof}[Proof of \Cref{thm:ebound1}]
Let $\mathbf{p}$ be any Nash equilibrium, and suppose that $S_1$ is the maximal tight subset. If $f(S_1)\geq 1$, then we are done, so suppose that $f(S_1)<1$. The Nash assumption then implies that any deviation by a queue in $S_1$ cannot decrease the value of each tight subset it is part of; note that we need the $f(S_1)<1$ assumption, as incentives are about the rates, and when $f(S_1)\geq 1$, the rate may remain 0 even if $f$ decreases.

For convenience, reindex and relabel so that $\vert S_1\vert = k$ and $S_1=[k]$. Fix any $\mathbf{x}\in \mathbb{R}^{m}_{\geq 0}$ such that $\sum_{i=1}^m x_i=k$. It suffices to show that
\begin{equation*}
    f([k]\vert \mathbf{p}) \geq \frac{\sum_{j=1}^m\mu_j (1-(1-x_j/k)^k)}{\sum_{i=1}^k \lambda_i}.
\end{equation*}
From now on, we omit the dependence on $\mathbf{p}$ in $f$ unless explicitly needed. 

Consider the level partition of $[k]$ guaranteed by \Cref{prop:levels} and suppose that there are $s$ levels. We continuously deform the Nash solution while monotonically decreasing $f([k])$, so that at the end of this process, we have a lower bound on the value of Nash. Given any strategy profile $\mathbf{p}$, we say a server $j$ is \emph{oversaturated} if $\sum_{i\in [k]} p_{ij}> x_j$ and \emph{undersaturated} if $\sum_{i\in [k]} p_{ij}< x_j$. We will continuously move probability mass from the queues from oversaturated to undersaturated servers. If no server is oversaturated, we will be done; notice that if a server is oversaturated, an easy averaging argument implies some server must be undersaturated. 

Suppose that there exists some oversaturated server. Let $i\in [k]$ be a queue at level-$s$, the top level. If $i$ nontrivially sends to an oversaturated server $j$ so that $p_{ij}>0$, we continuously decrease $p_{ij}$ and increase $p_{ij'}$ for some undersaturated server $j'$,
until either $j$ stops being oversaturated, $j'$ stops being undersaturated, or $p_{ij}$ hits zero. We claim that this deformation cannot increase $f([k])$. To see this, observe that because $\mathbf{p}$ is Nash, we know that any deviation by $i$ from one server it is nontrivially mixing at to another cannot increase $f(S)$ for all tight subsets $S$ containing $i$, hence there must be some $S\ni i$ such that $f(S)$ does not increase. But then \Cref{prop:levels} implies that some subset $S'$ containing all queues up to level-$s$ must have $f(S')$ not increase either. As $[k]$ is the only such subset, this deformation could not have actually increased $f([k])$.

Moreover, we claim that we can do this for all level-$s$ queues one-by-one 
without increasing $f([k])$. While the intermediate profiles are not Nash, because we only move probability mass from oversaturated to undersaturated servers, each oversaturated queue only has at most the same probability mass as it did at Nash while each undersaturated queue only has additional 
probability mass compared to what it had 
at Nash. 
As we have shown any such deviation by a level-$s$ server from an oversaturated queue to an undersaturated queue at Nash cannot increase $f([k])$, and now deviations are only worse at this intermediate stage while deforming the level-$k$ queue strategies, each such deformation still cannot increase $f([k])$. Therefore, we can continuously shift all probability mass from level-$s$ queues at oversaturated servers to undersaturated servers while never increasing $f([k])$.

Suppose we have now done this for all levels at least $\ell+1$ for some $\ell<s$ while not increasing $f([k])$, and we want to continue this process at level-$\ell$. Let $\mathbf{p'}$ be this intermediate strategy vector, where we note that for any queue $i$ below level $\ell+1$, $\mathbf{p'}_i=\mathbf{p}_i$. Again, if no server is oversaturated, we are done. Otherwise, suppose some queue $i$ at level-$\ell$ still sends to an oversaturated server $j$, and we again try to decrease $p'_{ij}$ and increase $p'_{ij'}$ for some undersaturated server $j'$ as before until the same stopping criterion. We must show that this too cannot increase $f([k])$.

Suppose otherwise that it did indeed increase $f([k])$ with respect to $\mathbf{p'}$. For a contradiction, it suffices to show that this implies that this same deviation, with respect to the original Nash solution $\mathbf{p}$, must have increased $f(S)$ for every subset $S$ containing all queues up to level-$\ell$. This is sufficient to obtain a contradiction as then \Cref{prop:levels} implies that \emph{every} subset containing $i$ improves at Nash with respect to this deviation, which violates the Nash property.

To prove this claim, let $S\subseteq [k]$ be any arbitrary tight subset at Nash containing all queues up to level-$\ell$. Because we assume that this deviation improves $f([k])$ with respect to $\mathbf{p'}$, it follows from taking partial derivatives that
\begin{equation*}
    \mu_j \prod_{r\in [k]\setminus \{i\}} (1-p'_{rj})< \mu_{j'} \prod_{r\in [k]\setminus \{i\}} (1-p'_{rj'}).
\end{equation*}
However, note that at $\mathbf{p'}$, as $j$ is still oversaturated, $p'_{rj}=0$ for all queues $r$ that are at strictly higher levels. As all queues at level-$\ell$ and below have $\mathbf{p'}_i=\mathbf{p}_i$ and $S$ contains all such queues, this inequality implies
\begin{equation*}
    \mu_j \prod_{r\in S\setminus \{i\}} (1-p_{rj})< \mu_{j'} \prod_{r\in [k]\setminus \{i\}} (1-p'_{rj'}).
\end{equation*}
Moreover, as $j'$ is undersaturated, we must have $p'_{rj'} \geq p_{rj'}$ for all $r\in [k]$ from the construction of this process, and removing terms in the product only increases the right side. Therefore, we deduce that
\begin{equation*}
    \mu_j \prod_{r\in S\setminus \{i\}} (1-p_{rj})< \mu_{j'} \prod_{r\in S\setminus \{i\}} (1-p_{rj'}).
\end{equation*}
This implies that this deviation also increases $f(S)$ at Nash. As $S$ was an arbitrary tight subset containing all queues up to level-$\ell$, the claim is proved. By the reduction described above, this is a contradiction, and therefore $f([k])$ must further decrease with respect to $\mathbf{p}'$. The argument extends analogously at all intermediate points of this process at level-$\ell$ by the same reasoning as before.

Therefore, by induction, it follows that we may continuously deform probability mass from oversaturated servers to undersaturated servers while only decreasing $f([k])$. At the end of this process, there cannot be any oversaturated servers, otherwise the process could have continued. In particular, if $\mathbf{p''}$ is the final probability vector at the end of this process, we have shown that $\sum_{i\in [k]} p''_{ij}=x_{j}$ for all servers $j$ and that $f([k]\vert \mathbf{p}'')\leq f([k]\vert \mathbf{p})$. We have
\begin{align*}
        f([k]\vert \mathbf{p})&=\frac{\sum_{j=1}^m \mu_j(1-\prod_{i=1}^k(1-p_{ij}))}{\sum_{i=1}^k \lambda_i} &&\text{(by definition)}\\
        &\geq \frac{\sum_{j=1}^m \mu_j(1-\prod_{i=1}^k(1-p''_{ij}))}{\sum_{i=1}^k \lambda_i} &&\text{(by construction)}\\
        &\geq \frac{\sum_{j=1}^m \mu_j(1-(1-x_j/k)^k)}{\sum_{i=1}^k \lambda_i} &&\text{(as symmetric profile maximizes product).}
    \end{align*}
The second inequality holds because given $\sum_{i\in [k]} p''_{ij}$, the maximizer of $\prod_{i\in [k]}(1-p''_{ij})$ is attained when each term is equal. As $\mathbf{x}$ was arbitrary, we may take the maximum of the right side over all $\mathbf{x}$ satisfying the constraints, and then the minimum over $k$. As $\mathbf{p}$ was an arbitrary Nash profile, this concludes the proof.
\end{proof}

\section{Discussion}
In this paper, we have studied a patient version of the queuing system of \cite{DBLP:conf/sigecom/GaitondeT20}; using careful probabilistic arguments to establish the incentive structure of the game, along with exploiting the analytic structure of the long-run rates of the induced Markov chain, we show that the correct bicriteria factor in this setting is $\frac{e}{e-1}$ via a novel deformation argument, strictly better than the factor of $2$ obtained in the no-regret learning setting. While the current result is not explicitly a learning result, this gap nevertheless suggests that no-regret behavior is not necessarily the correct notion of agent behavior in repeated games that carry strong interdependencies between rounds as in the priority structure here. It is an interesting question as to whether a natural form of learning can arrive at a Nash equilibrium of the patient version or at least results in stable outcomes without reaching an equilibrium. We note that empirical simulations of standard learning algorithms like EXP3 lead to time-independent, but myopic behavior as in \Cref{ex:ex1}. Moreover, though our restriction to time-independent policies exhibits quite rich behavior while enabling us to completely characterize the game-theoretic properties, perhaps there is a larger space of strategies where similar results hold. To that end, it may be necessary to explore the theoretical properties of more powerful learning algorithms in such settings that get the best of both worlds, namely balancing current rewards while maintaining long-run perspective. Whether such results are possible is an exciting open direction in resolving some of the deficiencies of traditional price of anarchy results; we leave a more systematic investigation of this direction to future work.

\bibliographystyle{alpha}
\bibliography{LearningQueues}

\newpage
\appendix

\section{Probability Tools}
We will use the following concentration results throughout the paper.

\begin{lem}[First Borel-Cantelli Lemma, Theorem 2.3.1 of \cite{durrett2019probability}]
\label{lem:borel}
Let $A_1,A_2,\ldots$ be a sequence of events with $\sum_{i=1}^{\infty}\Pr(A_i)<\infty$. Then with probability one at most finitely many of the $A_i$ occur.
\end{lem}

\begin{lem}[Azuma-Hoeffding]
\label{lem:azuma}
Let $\{\mathcal{F}_k\}_{k\leq n}$ be any filtration and let $A_k,B_k,\Delta_k$  satisfy the following conditions:
\begin{enumerate}
    \item $\Delta_k$ is $\mathcal{F}_k$-measurable and $\mathbb{E}[\Delta_k\vert \mathcal{F}_{k-1}]=0$. That is, the $\Delta_k$ form a martingale difference sequence.
    \item $A_k,B_k$ are $\mathcal{F}_{k-1}$-measurable and satisfy $A_k\leq \Delta_k\leq B_k$ almost surely.
\end{enumerate}
Then 
\begin{equation*}
    \Pr\left(\sum_{k=1}^n \Delta_k \geq t\right)\leq \exp\left(\frac{-2t^2}{\sum_{k=1}^n \|B_k-A_k\|_{\infty}}\right).
\end{equation*}
\end{lem}

\begin{lem}[Theorem 1 in \cite{witt2013fitness}]
\label{lem:geom1}
Let $X_1,\ldots,X_n$ be i.i.d. $\text{Geom}(\lambda)$ random variables, so that $\mathbb{E}[X_i]=\frac{1}{\lambda}$. Let $s=\frac{n}{\lambda^2}$. Then for all $\delta>0$,
\begin{equation*}
    \Pr\left( Z_n-\frac{n}{\lambda}<-\delta\right)\leq \exp\left(\frac{-\delta^2}{2s}\right),
\end{equation*}
and
\begin{equation*}
    \Pr\left( Z_n-\frac{n}{\lambda}>\delta\right)\leq \exp\left(\frac{-\delta}{4}\min\{\delta/s,\lambda\}\right)
\end{equation*}
where $Z_n:=\sum_{i=1}^n X_i$.
\end{lem}

\begin{proposition}[Corollary 6.2 of \cite{DBLP:conf/sigecom/GaitondeT20}, derived from \cite{witt2013fitness}]
\label{prop:geom3}
Let $\{G_{i,j}\}_{i\in [n],j\in [w]}$ be a family of independent geometric random variables such that for all $i,j$,
\begin{equation*}
    G_{i,j}\sim \text{Geom}(\lambda_i).
\end{equation*}
Let $Z_k^i=\sum_{j=1}^k G_{i,j}$. Then for any $\epsilon\in [0,1]$, 
\begin{equation}
\label{eq:geombound}
    \Pr\left(\exists i\in [n],j\in [w]: \bigg\vert Z_k^i-\frac{k}{\lambda_i}\bigg\vert \geq \frac{\epsilon w}{\lambda_i}\right)\leq 6n\exp\left(\frac{-\epsilon^2 w}{36}\right).
\end{equation}
\end{proposition}

\begin{thm}
[Theorem 1 in \cite{pemantle1999moment}]
\label{thm:pemantle}
Let $X_1,X_2,\ldots$ be a sequence of nonnegative random variables with the property that
\begin{enumerate}
    \item There exists constants $\alpha,\beta>0$ such that if $x_n>\beta$, then 
    \begin{equation*}
        \mathbb{E}[X_{n+1}-X_n\vert \mathcal{F}_n]<-\alpha,
    \end{equation*}
    where $\mathcal{F}_n=\sigma(X_0,\ldots,X_n)$ is the history until period $n$ and $X_n=x_n$.
    
    \item There exists $p>2$ and $\theta>0$ a constant such that for any history, 
    \begin{equation*}
        \mathbb{E}[\vert X_{n+1}-X_n\vert^p\vert \mathcal{F}_n]\leq \theta.
    \end{equation*}
\end{enumerate}
Then, for any $0<r<p-1$, there exists an absolute constant $M=M(\alpha,\beta,\theta,p,r)$ not depending on $n$ such that $\mathbb{E}[X_n^r]\leq M$ for all $n$.
\end{thm}

\section{Table of Notation}
\label{sec:table}
For convenience, we repeat the definitions of various expressions considered in the main text.
\renewcommand{\arraystretch}{1.2}
\begin{center}
\begin{tabular}{c c p{2.5in}}
\hline
Symbol & Formula & Definition \\
\hline\hline
   $\bm{\lambda}$  &  & Vector of length $n$ of queue arrival rates in descending order.\\
   \hline
   $\Delta^{m-1}$ & & Probability simplex over $m$ element set.\\
   \hline
   $\bm{\mu}$  &  & Vector of length $m$ of server success rates in descending order.\\
   \hline
   $\mathbf{p}$ & & Vector of queue randomizations over servers  in $(\Delta^{m-1})^n$.\\
   \hline
   $\widetilde{T_t^i}$ & & Timestamp of oldest packet at queue $i$ at time step $t$.\\
   \hline
   $T_t^i$ & $ \max\{0,t-\widetilde{T_t^i}\}$ & Age of queue $i$ at time $t$.\\
   \hline
   $\alpha(S\vert \mathbf{p},\mathbf{\mu},S')$ & $\sum_{j=1}^m \mu_j\prod_{i\in S'}(1-p_{i,j})(1-\prod_{i\in S}(1-p_{i,j}))$ & Expected number of packets cleared by queues in $S$ if all have packets in a round and have priority over all queues except for those in $S'$ and each such queue also has packets in the round.\\
   \hline
   $\lambda(S)$ & $\sum_{i\in S}\lambda_i$ & Sum of arrival rates of queues in $S$.\\
   \hline
   $f(S\vert \mathbf{p},\bm{\mu},\bm{\lambda},S')$ & $\alpha(S\vert \mathbf{p},\mathbf{\mu},S')/\lambda(S)$ & Ratio of expected packets cleared by $S$ with priority over all queues except $S'$ to total arrival rate of $S$.\\
   \hline
   $S_k(\mathbf{p},\bm{\mu},\bm{\lambda})$ & & $k$th subset output in the algorithm of \Cref{section:basicproperties}.\\
   \hline
   $U_k(\mathbf{p},\bm{\mu},\bm{\lambda})$ & $\cup_{\ell=1}^k S_{\ell}(\mathbf{p},\bm{\mu},\bm{\lambda})$ & Union of first $k$ outputted subsets in the algorithm of \Cref{section:basicproperties}.\\
   \hline
   $r_i(\mathbf{p},\bm{\mu},\bm{\lambda})$ & & Outputted aging rate of queue $i$ in the algorithm of \Cref{section:basicproperties}.\\
   \hline
   $f_k(\mathbf{p},\bm{\mu},\bm{\lambda})$ & $f(S_k(\mathbf{p},\bm{\mu},\bm{\lambda})\vert \mathbf{p},\bm{\mu},\bm{\lambda},U_{k-1})$ & Value of $f$ for $S_k$ when $U_{k-1}$ has priority.\\
   \hline
   $g_k(\mathbf{p},\bm{\mu},\bm{\lambda})$ & $\max\{0,1-f_k(\mathbf{p},\bm{\mu},\bm{\lambda})\}$ & Outputted rate for $S_k$; equivalently, value of $r_i$ for any $i\in S_k$.

\end{tabular}
\end{center}
When the values of, or dependencies on, $\mathbf{p},\bm{\mu},\bm{\lambda}$ are clear from context, we omit them for notational ease.

\section{Proofs Deferred from \Cref{section:convergence}}

\label{sec:appendixconv}
This section is devoted to the remaining details of the intermediate claims in the proof of \Cref{thm:gamemain}. Our main technical result of this section asserts that with high probability, the maximum queue age increases at a rate of at most $(1-(1-\epsilon)\cdot f_1)$ on the next $w$ steps for a large enough $w$. In fact, more generally, the following holds:
\begin{proposition}[\Cref{prop:upperbound}, restated]
Fix $\epsilon>0$. For any integer $a\in \mathbb{N}$, let $w=a\cdot \lceil \frac{6}{\epsilon}\rceil^{n-1}$. Suppose it holds at time $t$ that $\max_{i\in [n]} T^i_t\geq w\cdot f_1$. Then 
\begin{equation*}
    \max_{i\in [n]} T^i_{t+w} -\max_{i\in [n]} T^i_t \leq (1-(1-\epsilon)\cdot f_1)\cdot w
\end{equation*}
with probability at least $1-C_1\exp(-C_2a)$, where $C_1,C_2>0$ are absolute constants depending only on $n,\epsilon,\bm{\lambda}, \bm{\mu},\mathbf{p}$, but not on $a$.

More generally, for each $s\geq 1$, if $\max_{i\not\in U_{s-1}} T^i_t\geq w\cdot f_s$, then 
\begin{equation*}
    \max_{i\not\in U_{s-1}} T^i_{t+w} -\max_{i\not\in U_{s-1}} T^i_t \leq (1-(1-\epsilon)\cdot f_s)\cdot w
\end{equation*}
with probability at least $1-C_1\exp(-C_2a)$, where $C_1,C_2>0$ are absolute constants depending only on $n,\epsilon,\bm{\lambda}, \bm{\mu},\mathbf{p}$, but not on $a$.
\end{proposition}

Note that the first part is simply the $j=1$ case of the more general statement. This proposition will follow from the following lemma that will prove 
inductively:
\begin{lem}
\label{lem:decreaseage}
Fix $s\geq 1$ and $\epsilon>0$. Then, for each $1\leq \tau\leq n$ and all $a\in \mathbb{N}$, the following holds: let $w=a\cdot \lceil \frac{6}{\epsilon}\rceil^{\tau-1}$, and suppose that at time $t$, $M^*:=\max_{i\not\in U_{s-1}} T^i_t\geq w\cdot f_s$, and that the set 
\begin{equation*}
    I=\{ i\not \in U_{s-1}: T^i_t\geq M^*-w\cdot f_s\}
\end{equation*}
has $\vert I\vert\leq \tau$. Then, with probability at least $1-C_1\exp(-C_2a)$, we have
\begin{equation*}
    \max_{i\not\in U_{s-1}} T^i_{t+w}- \max_{i\not\in U_{s-1}} T^i_t \leq (1-(1-\epsilon)\cdot f_s)\cdot w,
\end{equation*}
where $C_1,C_2>0$ are absolute constants depending only on $n,\epsilon,\bm{\lambda}, \bm{\mu},\mathbf{p}$, but not $a$.
\end{lem}
The proposition follows as for any $s\geq 1$ and $a\in \mathbb{N}$, $\vert I\vert\leq n$.\\

We now turn to proving the lemma inductively. The case for $\tau=1$ will turn out to be relatively easy; indeed, this case just says that there is a single very old queue among those not in $U_{s-1}$, so we will be able to lower bound the number of packets she clears by simply assuming every queue in $U_{s-1}$ is older than her. To extend this to higher $\tau$ will be more difficult. To do this, we will chunk together many windows that we know have this property for smaller values of $\tau$ and then leverage two facts to get a win-win situation. We will be able to easily show that at least one queue in $I$ is always decreasing at the correct rate; then if all queues in $I$ are ``close,'' they also are at the correct rate. If not, then on the next chunk, they will clear at the correct rate on the next chunk inductively with high probability. If we do this chunking properly, we can provide good control on all the queues on every chunk via conditioning and using one of these two cases.

We now carry out this high-level plan. Fix $s\geq 1$. For reference, we will use the following notation and conventions to ``organize randomness'' in a convenient way to analyze the process:
\begin{enumerate}
    \item $w$ will denote a given \emph{window} length composed of $B$ consecutive \emph{blocks} of $L$ \emph{steps}, so that $w = B\cdot L$. In the following, we will make this convention, and because we will be considering the behavior of the process on some fixed window, we may as well reindex $t=1$ for convenience, so that each window we consider will go from $t=1$ to $w$. We will reserve the superscript $t=0$ to denote the value of the ages at the very beginning of the window we consider.
    
    \item Recall we use the shorthand $f_s:=f(S_s\vert U_{s-1})$ and $g_s = \max\{0,1-f_s\}$.
    
    \item With this convention, we will often define (and will make clear from context) at the beginning of some considered window of fixed length $w$, fixed $s\geq 1$, and fixed $\epsilon>0$:
    \begin{gather*}
    M^* :=\max_{i\in [n]\setminus U_{s-1}} T^i_0\\
        T^*:=M^* -(1-\epsilon)\cdot w\cdot f_s.
    \end{gather*}
    We will often refer to $T^*$ as the \emph{target} value for this window, which does not change over the course of the window (notice it is measured at the beginning of the window).
    Then, define
    \begin{equation*}
        I=\{i\notin U_{s-1}: T^i_0\geq M^*-w\cdot f_s\}.
    \end{equation*}
    That is, $I$ is the set of queues whose age is within $w\cdot f_s$ of the oldest age, measured at the beginning of the window. Our goal will be to eventually show if $w$ is sufficiently large, then with high probability, every queue in $I$ has age below $T^*$ at the end of the next $w$ steps before accounting for $w$ steps of aging, and of course all queues not in $I$ are already strictly below $T^*$ by definition. This will imply that the maximum age grows by at most $(1-(1-\epsilon)\cdot f_s)\cdot w$ once we account for the $w$ steps of aging over this window.
    
    \item Given a window of length $w=B\cdot L$, $\mathcal{F}^{(b)}_{\ell}$ is the filtration of $\sigma$-algebras generated up to step $\ell$ in the $b$th block, for $b=1,\ldots,B$. In particular,
    \begin{equation*}
        \mathcal{F}_0^{(1)}\subseteq \mathcal{F}_1^{(1)}\subseteq \ldots\subseteq \mathcal{F}_{L}^{(1)}\subseteq \mathcal{F}_{0}^{(2)}\subseteq \ldots\subseteq \mathcal{F}_{\ell}^{(b)}\subseteq\ldots\subseteq \mathcal{F}_{L}^{(B)}.
    \end{equation*}
    
    \item $X_{i,\ell}^{(b)}$ will be the indicator that queue $i$ cleared a packet in timestep $\ell$ in the $b$th block. Note that $X_{i,\ell}^{(b)}$ is $\mathcal{F}^{(b)}_{\ell}$-measurable.
    \item $Y_{i,\ell}^{(b)}$ will be a sequence of random variables, with same interpretation of the indices, defined as follows: for $b,\ell$ such that every queue in $I$ is still above $M^*-w\cdot f_s$ at the start of the $\ell$th step of the $b$th block, set $Y_{i,\ell}^{(b)}=X_{i,\ell}^{(b)}$. If this does not hold for some $b,\ell$, then let the $Y_{i,\ell}^{(b)}$ be arbitrary indicator random variables satisfying 
    \begin{equation*}
        \mathbb{E}\left[\sum_{i\in I} Y_{i,\ell}^{(b)}\bigg\vert \mathcal{F}_{\ell-1}^{(b)}\right]\geq f_s\cdot \left(\sum_{i\in I} \lambda_i\right).
    \end{equation*}
    Note that $Y_{i,\ell}^{(b)}$ is $\mathcal{F}^{(b)}_{\ell}$-measurable.
    
    These random variables are purely for technical convenience because they have an \emph{a priori} lower bound on the conditional expectation, which is not always true of $X$ (if for instance, queues have already cleared a lot and thus some queues in $I$ have lost priority over those not in $I$). These have no real meaning and are just convenient for the analysis. 
    
    \item $G_{i,\ell}^{(b)}$ will be i.i.d. Geom($\lambda_i$) random variables in the same way for $\ell\in [L],b\in [B]$. We define the partial sums on each window $Z_{i,k}^{(b):= \sum_{\ell=1}^k G_{i,\ell}^{(b)}}$. 
    
    For each $b=1,\ldots,B$, we make the convention that the $G_{i,\ell}^{(b)}$ are sampled for all $i\in [n]$ and $\ell\in [L]$ at the beginning of the $b$th block, so that they are all $\mathcal{F}_0^{(b)}$ measurable; there is no corresponding queuing step. When queue $i$ clears her $k$th packet in the $b$th block, her timestamp decreases by $G_{i,k}^{(b)}$; equivalently, when this happens, her timestamp will have decreased on the $b$th block by $Z_{i,k}^{(b)}$ so far.
\end{enumerate}

We now proceed with the proof:

\begin{proof}[Proof of \Cref{lem:decreaseage}]
Fix any group $s\geq 1$ and $\epsilon>0$. Note that if $f_s=0$ (i.e. queues never clear packets), then the lemma holds trivially with probability 1 as every queue in this subset increases deterministically by $w$ in age on any window of length $w$, so we suppose otherwise. The proof is by induction on $\tau=\vert I\vert$.\\

\textbf{Base Case: $\tau=1$ (One Old Queue)}: We first consider $\tau=1$. Let $a\in \mathbb{N}$ be arbitrary, and then set $w=a$. If $i^*\in I$ is the unique queue in the subset, then $i^*=\arg\max_{i\in [n]\setminus U_{s-1}} T_0^i$. We must show that at the end of $w$ steps, this queue has decreased age (before accounting for aging) by at least $(1-\epsilon)\cdot w\cdot f_s$ with high probability; the desired conclusion then follows from adding back in the $w$ steps of aging. For simplicity, write $\lambda=\lambda_{i^*}$ as the arrival rate of this unique queue. We prove the claim in the following two steps:
\begin{enumerate}
    \item \textbf{First, we show that with high probability, the number of packets this queue must clear on this window to get below the target is not too large.}
    
    Recall that to model this random process on the next $w$ steps, let $G_1,\ldots,G_w$ be i.i.d. Geom($\lambda$) random variables (write $w=1\cdot w$ to indicate our window is just one block of $w$ steps, so we omit the superscripts for blocks).  As usual, write $Z_{k}=\sum_{\ell=1}^k G_{\ell}$. We model the next $w$ steps for this queue by saying that when she clears her $k$th packet on this window, her timestamp decreases by $G_k$; equivalently, clearing $k$ packets decreases her timestamp by $Z_k$ collectively. Sample all of these random variables before hand, and let $\mathcal{F}_0$ be the $\sigma$-algebra generated by the previous history, as well as these random variables. $\mathcal{F}_{\ell}$ will be the filtration generated by all prior events up to the $\ell$th step in this window.

Next, define the random variable
\begin{equation*}
    K^*=\min\{k\in [w]: Z_k\geq (1-\epsilon)\cdot w\cdot f_s\},
\end{equation*}
with the convention that if this set is empty, then $K^*=w+1$. Observe that $K^*$ is exactly the number of packets that this queue must clear to clear the target.

We claim that with high probability, $K^*\leq K:=\lambda\cdot (1-\epsilon/2)\cdot w\cdot f_s$. To see this, apply \Cref{lem:geom1} with $K=(1-\epsilon/2)\cdot w\cdot f_s$ and $\delta=\frac{wf_s\epsilon}{2}$ to see that for this choice of $K$,
\begin{align}
    \Pr(K^*>K)&=\Pr\left(Z_K < (1-\epsilon)\cdot w\cdot f_s\right)\\
    &=\Pr\left(Z_K< \frac{\lambda\cdot (1-\epsilon/2)\cdot w\cdot f_s}{\lambda}-\frac{\epsilon\cdot w\cdot f_s}{2}\right)\\ &\leq\exp\left(\frac{-w^2f_s^2\lambda^2\epsilon^2 }{8\lambda wf_s(1-\epsilon/2)}\right)\\
    \label{eqn:KBound}
    &\leq \exp\left(\frac{-\lambda wf_s\epsilon^2 }{8}\right).
\end{align}

Recall that $f_s>0$; therefore, we have shown that with high probability, the queue will only need to clear at most $K$ packets to get below the desired target.
    \item \textbf{Next, we show that with high probability, this queue will clear at least as many packets as required by the previous claim.}
    
    Let $X_1,\ldots,X_w$ be the indicator variables for if the queue cleared a packet on each step of the window. Note that these random variables are path-dependent, but $X_{\ell}$ is $\mathcal{F}_{\ell}$-measurable. Then the queue's timestamp decreases by at least $(1-\epsilon)\cdot w\cdot f_s$ if and only if $\sum_{\ell=1}^w X_{\ell}\geq K^*$ by definition. Therefore, the probability the queue's timestamp decreases by at least $(1-\epsilon)\cdot w\cdot f_s$ on the next $w$ steps is
\begin{equation*}
    \Pr\left(\sum_{\ell=1}^w X_{\ell}\geq K^*\right)=\sum_{k=1}^w \Pr\left(\sum_{\ell=1}^w X_{\ell}\geq K^*\bigg\vert K^*=k\right)\Pr(K^*=k)
\end{equation*}
Consider the family $Y_1,\ldots,Y_w$ of indicator random variables that we couple with $X_1,\ldots,X_w$ as follows: while $\sum_{q=1}^{\ell-1} X_q < K^*$, set $Y_{\ell}=X_{\ell}$. While $\sum_{q=1}^{\ell-1} X_q \geq K^*$ on a sample path, let $Y_{\ell}$ be an arbitrary indicator random variable satisfying $\mathbb{E}[Y_{\ell}\vert \mathcal{F}_{\ell-1}]\geq \lambda\cdot  f_s$. Notice that by construction, we always have $\mathbb{E}[Y_{\ell}\vert \mathcal{F}_{\ell-1}]\geq \lambda\cdot f_s$; this is because if $\sum_{q=1}^{\ell-1} X_{q}<K^*$, then the queue is still above the target, and therefore by assumption has priority. As this queue is the oldest not in $U_{s-1}$, she has priority over all other queues. If $V\subseteq U_{s-1}$ is some subset of queues with priority over her before she reaches her target, we know that in this case,

\begin{align*}
    \mathbb{E}[Y_{\ell}\vert \mathcal{F}_{\ell-1}]&= \mathbb{E}[X_{\ell}\vert \mathcal{F}_{\ell-1}]\\
    &= \lambda\cdot f(\{i^*\}\vert V)\\
    &\geq \lambda\cdot f(\{i^*\}\vert U_{s-1})\\
    &\geq \lambda\cdot f_s,
\end{align*}
where we use set monotonicity and the fact that $f_s$ is the minimal value of $f(\cdot \vert U_{s-1})$ over all subsets contained in the complement.

 Of course, if $\sum_{q=1}^{\ell-1} X_q\geq K^*$, $\mathbb{E}[Y_{\ell}\vert \mathcal{F}_{\ell-1}]\geq \lambda f_s$ simply by construction. However, because $X_{\ell}$ and $Y_{\ell}$ are equal while the queue is above the target, or equivalently before having cleared $K^*$ packets, it follows that the events that $\sum_{\ell=1}^w X_{\ell}\geq K^*$ and $\sum_{\ell=1}^w Y_{\ell}\geq K^*$ have the same probability. Recall again that $K:=\lambda\cdot (1-\epsilon/2)\cdot w\cdot f_s$. We obtain the probability the queue's timestamp gets below the target is at least
\begin{align*}
    \sum_{k=1}^K \Pr \left(\sum_{\ell=1}^w X_{\ell}\geq k \bigg\vert K^*=k\right)  &\Pr(K^*=k)= \sum_{k=1}^K \Pr\left(\sum_{\ell=1}^w Y_{\ell}\geq k \bigg\vert K^*=k\right)\Pr(K^*=k)\\
    &\geq \sum_{k=1}^K \Pr\left(\sum_{\ell=1}^w Y_{\ell}\geq K \bigg\vert K^*=k\right)\Pr(K^*=k)\\
    &\geq \sum_{k=1}^K \Pr\left(\sum_{\ell=1}^w Y_{\ell}\geq \sum_{\ell=1}^w \mathbb{E}[Y_{\ell}\vert \mathcal{F}_{\ell-1}]-\frac{\lambda wf_s\epsilon}{2}\bigg\vert K^*=k\right)\Pr(K^*=k)\\
    &\geq \left(1-\exp\left(-\lambda^2wf_s^2\epsilon^2/4\right)\right)\Pr(K^*\leq K)\\
    &\geq \left(1-\exp\left(-\lambda^2wf_s^2\epsilon^2/4\right)\right)\left(1-\exp\left(\frac{-\lambda wf_s\epsilon^2 }{8}\right)\right)\\
    &\geq 1-2\exp\left(\frac{-\lambda^2wf_s^2\epsilon^2}{8}\right).
\end{align*}

We use Azuma-Hoeffding in the fourth line applied to $\Delta_{\ell}:=Y_{\ell}-\mathbb{E}[Y_{\ell}\vert \mathcal{F}_{\ell-1}]$, which we note is surely between $-1$ and $1$.
\end{enumerate}
As $w=a=a\cdot \lceil \frac{6}{\epsilon}\rceil^{1-1}$, it is clear that the probability this occurs is of the claimed form. To be safe, one should take $\lambda=\min_i \lambda_i$ to make the bound hold uniformly, independent of the identity of this queue.\\

\noindent\textbf{Inductive Step for $\tau>1$}
Suppose that the proposition holds up to $\tau$, and now we must show it holds for $\tau+1$. Let $a\in \mathbb{N}$, and then $w=a\cdot \lceil \frac{6}{\epsilon}\rceil^{\tau}=\lceil \frac{6}{\epsilon}\rceil\cdot \big(a\cdot \lceil \frac{6}{\epsilon}\rceil^{\tau-1}\big)$. In our notation, we have $w=B\cdot L$, where $B=\lceil \frac{6}{\epsilon}\rceil$ and $L=a\cdot \lceil \frac{6}{\epsilon}\rceil^{\tau-1}$.

Define as usual \begin{gather*}
    M^* = \max_{i\in [n]\setminus U_{s-1}} T^i_0\\
    T^*:=M^* -(1-\epsilon)\cdot w\cdot f_s\\
    I=\{i\in [n]\setminus U_{s-1}: T^i_t\geq M^*-w\cdot f_s\},
\end{gather*}
and suppose now that $\vert I\vert\leq  \tau+1$ and $M^*\geq w\cdot f_s$. For all $i\in I$, $\ell=1,\ldots,L$, and $b=1,\ldots,B$, define the random variables $X_{i,\ell}^{(b)}$ and $Y_{i,\ell}^{(b)}$ as described above, as well as the $\sigma$-algebras $\mathcal{F}_{\ell}^{(b)}$.

First, note that for any $b,\ell$, if no queue in $I$ has timestamp below $M^*-w\cdot f_s$ at the $\ell$th step of the $b$th block (without accounting for aging), then every queue in $I$ has priority over queues not in $I$. Then for the same reason as in the base case we have
\begin{equation*}
    \mathbb{E}\left[\sum_{i\in I}X_{i,\ell}^{(b)}\bigg\vert \mathcal{F}_{\ell-1}^{(b)}\right]\geq f_s\cdot \left(\sum_{i\in I} \lambda_i\right).
\end{equation*}
Then by construction, we always have
\begin{equation*}
   \mathbb{E}\left[\sum_{i\in I}Y_{i,\ell}^{(b)}\bigg\vert \mathcal{F}_{\ell-1}^{(b)}\right]\geq f_s\cdot \left(\sum_{i\in I} \lambda_i\right),
\end{equation*}
and this holds regardless of the conditioning. In particular, it follows that for every $b=1,\ldots,B$, we have
\begin{equation*}
    \sum_{\ell=1}^{L}\mathbb{E}\left[\sum_{i\in I}Y_{i,\ell}^{(b)}\bigg\vert \mathcal{F}_{\ell-1}^{(b)}\right]\geq L\cdot f_s\cdot \left(\sum_{i\in I} \lambda_i\right).
\end{equation*}

Now, we define the following events for $b=1,\ldots,B$:
\begin{gather*}
    D_b = \left\{ T^i_{b\cdot L}\leq M^*-(b-3)\cdot L\cdot f_s(1-\epsilon/2),\forall i\in I\right\}\\
    E_b=\left\{Z_{i,\ell}^{(b)}\geq \frac{\ell}{\lambda}-\frac{\epsilon L f_s}{4},\forall i\in I,\ell=1,\ldots, L\right\}\\
    F_b = \left\{ \sum_{\ell=1}^{L}\sum_{i\in I} Y_{i,\ell}^{(b)}\geq L\cdot f_s(1-\epsilon/4)\left(\sum_{i\in I} \lambda_i\right)\right\}.
\end{gather*}

We make the convention that $D_0,E_0,F_0$ are just the trivial event $\Omega$ that happens surely. Then define inductively for $b<B$:
\begin{gather*}
    A_0 =\Omega\\
    A_{b+1}=A_b\cap D_{b+1}\cap E_{b+1}\cap F_{b+1}.
\end{gather*}
We note that $A_B$ implies $D_B$ by construction so that if $A_B$ holds, then at the end of this window of $w$ steps (and again without accounting for $w$ steps of aging), for all $i\in I$
\begin{equation*}
    T^i_w=T^i_{B\cdot L}\leq M^*-(B-3)\cdot L\cdot f_s(1-\epsilon/2).
\end{equation*}
Recall that our goal is that $T^i_w\leq M^*-B\cdot L\cdot f_s(1-\epsilon)$ (before accounting for $w$ steps of aging); note that our choice of $B=\lceil \frac{6}{\epsilon}\rceil$ satisfies
\begin{equation*}
    (B-3)\cdot L\cdot f_s\cdot(1-\epsilon/2)\geq B\cdot L\cdot f_s\cdot (1-\epsilon).
\end{equation*}

Now, we lower bound the probability that $A_{b}$ holds inductively. As $A_{b+1}\subseteq A_b$, we have
\begin{align*}
    \Pr(A_{b+1})&=\Pr(D_{b+1}\cap E_{b+1}\cap F_{b+1}\vert A_b)\Pr(A_b)\\
    &\geq (\Pr(D_{b+1}\vert A_b)-\Pr(E_{b+1}^c\cup F_{b+1}^c\vert A_b))\Pr(A_b),
\end{align*}
where we just use a union bound. By a union bound and \Cref{prop:geom3}, 
\begin{equation*}
    \Pr(E_{b+1}^c\vert A_b)\leq 6(\tau+1)\exp\left(\frac{- L\cdot \lambda^2f_s^2\epsilon^2 }{144}\right),
\end{equation*}
where $\lambda=\min_{i\in [n]} \lambda_i$. A familiar argument by Azuma-Hoeffding gives
\begin{equation*}
    \Pr(F_{b+1}^c\vert A_b)\leq \exp\left(-\left(\sum_{i\in I} \lambda_i\right)^2L f_s^2\epsilon^2/16n\right)\leq \exp(-\lambda^2 L f_s^2\epsilon^2/16n),
\end{equation*}
where we note that $\sum_{i\in I} Y_{i,\ell}^{(b+1)}-\mathbb{E}\left[\sum_{i\in I} Y_{i,\ell}^{(b+1)}\right]$ is surely between $-n$ and $n$ for that extra factor.
Therefore, a union bound implies
\begin{equation*}
   \Pr(A_{b+1})\geq \left(\Pr(D_{b+1}\vert A_b)-6(\tau+2)\exp\left(\frac{-\lambda^2 L f_s^2\epsilon^2 }{144n}\right)\right)\Pr(A_b)
\end{equation*}

We now show that $\Pr(D_{b+1}\vert A_b)$ is large using a case analysis:
\begin{enumerate}
    
    \item \textbf{Case 1: No Gap} First suppose that after the $b$th block, there is no large gap between the maximum and minimum timestamp in $I$; that is, (without accounting for aging, which affects queues equally)
    \begin{equation*}
        \max_{i\in I} T^i_{b\cdot L}-\min_{i\in I} T^i_{b\cdot L}\leq  L\cdot f_s(1-\epsilon/2).
    \end{equation*}
    We show that this, along with the other assumptions in $A_b$, already imply the event $D_{b+1}$, so there is no need to analyze what happens on the $b+1$th block.
    
    Note that because there is no large gap, $D_{b+1}$ will be implied by
    \begin{equation}
    \label{eq:nogap}
        \min_{i\in I} T^i_{b\cdot L}\leq  M^*-b\cdot L\cdot f_s (1-\epsilon/2),
    \end{equation}
    because then (again, as we only account for aging at the end)
    \begin{equation*}
        \max_{i\in I} T^i_{(b+1)\cdot L}\leq \max_{i\in I} T^i_{b\cdot L}\leq  M^*-(b-1)\cdot L\cdot f_s(1-\epsilon/2).
    \end{equation*}
    We now show \Cref{eq:nogap} is indeed implied by $A_b$, which we recall implies $E_b$ and $F_b$ for all $q\leq b$. This means that for each $q\leq b, i\in I, \ell=1,\ldots, L$
    \begin{equation}
    \label{eqn:partsums}
        Z_{i,\ell}^{(q)}\geq \frac{\ell}{\lambda}-\frac{\epsilon L f_s}{4},
    \end{equation}
    as well as
    \begin{equation*}
        \sum_{q=1}^b\sum_{\ell=1}^{L}\sum_{i\in I} Y_{i,\ell}^{(q)}\geq b\cdot L\cdot f_s(1-\epsilon/4)\left(\sum_{i\in I} \lambda_i\right).
    \end{equation*}
    If $\min_{i\in I} T_{b\cdot L}^i\leq M^* - B\cdot L\cdot f_s$, we are done by the above discussion, so suppose not; this means no queue in $I$ has timestamp below $M^*-wf_s$ (before accounting for aging). By our coupling, we have $X_{i,\ell}^{(q)}=Y_{i,\ell}^{(q)}$ up until now as no queue in $I$ has lost priority to those outside of $I$, so this last inequality is equivalent to
    \begin{equation*}
        \sum_{q=1}^b\sum_{\ell=1}^{L}\sum_{i\in I} X_{i,\ell}^{(q)}\geq b\cdot L\cdot f_s(1-\epsilon/4)\left(\sum_{i\in I} \lambda_i\right).
    \end{equation*}
    By averaging, this implies that there exists some $i\in I$ such that
    \begin{equation*}
        \sum_{q=1}^{b}\sum_{\ell=1}^{L} X_{i,\ell}^{(q)}\geq b\cdot \lambda_i\cdot L\cdot f_s(1-\epsilon/4).
    \end{equation*}
    By (\ref{eqn:partsums}), this implies this queue $i$ has decreased in age by at least
    \begin{equation*}
        b\cdot L\cdot f_s(1-\epsilon/4)-\frac{b\cdot \epsilon \cdot L f_s}{4}=b\cdot L\cdot f_s(1-\epsilon/2).
    \end{equation*}
    Because all queues in $I$ have timestamp at most $M^*$ at the beginning of this process, this implies that (before accounting for aging)
    \begin{equation*}
        \min_{i\in I} T^i_{b\cdot L}\leq  M^*-b\cdot L\cdot f_s(1-\epsilon/2),
    \end{equation*}
    which combined with the no-gap assumption implies $D_{b+1}$. Therefore, when there is no gap, the probability of $D_{b+1}$ conditioned on $A_b$ is $1$.
    
    \item \textbf{Case 2: Large Gap} Suppose after the $b$th block, there is a large enough gap between the maximum and minimum timestamp in $I$; specifically, suppose that (again, without aging) 
    \begin{equation*}
        \max_{i\in I} T_{b\cdot L}^i-\min_{i\in I} T_{b\cdot L}^i> L\cdot f_s(1-\epsilon/2).
    \end{equation*}
    The previous case showed that $A_b$ itself implies
    \begin{equation*}
        \min_{i\in I} T^i_{b\cdot L}\leq  M^*-b\cdot L\cdot f_s(1-\epsilon/2)
    \end{equation*}
    As we have conditioned on $A_b$, $\max_{i\in I}T_{b\cdot L}^i\leq M^*-(b-3)\cdot L f_s(1-\epsilon/2)$; if this held instead with $b-2$, we would already be done. If not, as we have set $L=a\cdot \lceil \frac{6}{\epsilon}\rceil^{\tau-1}$ and evidently the set of queues with timestamp within $L\cdot f_s$ of the maximum in $I$ has size at most $\tau$, we may apply the inductive hypothesis to see that with probability at least $1-C_1\exp(-C_2a)$, for some absolute constants $C_1$ and $C_2$ depending only on $n,\epsilon,\bm{\lambda}, \bm{\mu},\mathbf{p}$ that every such queue decreases by at least $L\cdot f_s\cdot(1-\epsilon/2)$ on this block by our choice of $L$; therefore, conditioned on $A_b$, with probability at least $1-C_1\exp(-C_2a)$, these queues decrease by enough to satisfy $D_{b+1}$.
\end{enumerate}

Combining these cases, we conclude that
\begin{equation*}
    \Pr(A_{b+1})\geq \left(1-C_1\exp(-C_2a)-6(\tau+2)\exp\left(\frac{-\lambda^2 L f_s^2\epsilon^2 }{144n}\right)\right)\Pr(A_b)
\end{equation*}
 Unravelling this recurrence and using $\Pr(A_0)=1$, we conclude that
\begin{equation*}
    \Pr(A_m)\geq 1-B(C_1\exp(-C_2a))-6B(\tau+2)\exp\left(\frac{-\lambda^2 L f_s^2\epsilon^2 }{144n}\right).
\end{equation*}
As $L=a\cdot \lceil \frac{6}{\epsilon}\rceil^{\tau-1}$ and $B=\lceil \frac{6}{\epsilon}\rceil$, this evidently has the claimed form of $1-C_1'\exp(-C_2'a)$ for absolute constants $C_1',C_2'>0$ depending only on system parameters and not $a$, completing the inductive step. With this, the lemma is proved.
\end{proof}

For this proposition to yield anything useful, we will need a corresponding lower bound that asserts roughly that if groups have separated according to what the algorithm asserts, then the aging rate of the \emph{average} queue in a group grows at the conjectured rate.

\begin{proposition}[\Cref{prop:lowerbound}, restated]
For any $s\geq 1$ and any fixed $\epsilon>0$, the following holds: suppose that at time $t$, it holds that
\begin{equation*}
    \min_{i\in U_s} T^i_t -\max_{i\in S_{s+1}} T^i_t\geq  2\cdot \frac{w}{\lambda_n}.
\end{equation*}
Then with probability $1-A\exp(-Bw)$ where $A,B>0$ are absolute constants not depending on $w$, we have
\begin{equation*}
    \sum_{i\in S_{s+1}} \lambda_i T^i_{t+w}-\sum_{i\in S_{s+1}} \lambda_iT^i_t\geq (1-(1+\epsilon)f_{s+1})\cdot w \cdot \left(\sum_{i\in S_{s+1}} \lambda_i\right).
\end{equation*}

Moreover, for any fixed $\epsilon>0$, with probability at least $1-A\exp(-Bw)$ it holds that
\begin{equation*}
    \sum_{i\in S_1} \lambda_iT^i_w \geq (1-(1+\epsilon)f_1)\cdot w\cdot \left(\sum_{i\in S_1} \lambda_i\right).
\end{equation*}
\end{proposition}
\begin{proof}
We prove the second statement first, as it is slightly simpler and the main idea will reappear. 

Using similar notation to the last proof, let $X_{i,t}$ denote the indicator variable that queue $i$ cleared a packet at time $t$. Similar to as in that proof, we have for every $t\geq 1$,
\begin{equation*}
    \mathbb{E}\left[\sum_{i\in S_1} X_{i,t}\bigg\vert \mathcal{F}_{t-1}\right]\leq \left(\sum_{i\in S_1} \lambda_i\right)\cdot f_1.
\end{equation*}
Indeed, other queues may have priority at time $t$ and some queues in $S_1$ may be empty at time $t$, so this is just an upper bound.

Again, let $G_{i,\ell}$ be i.i.d. geometric random variables with parameter $\lambda_i$ for $\ell=1,\ldots,w$. In a now familiar argument, the interpretation is that when queue $i$ clears her $\ell$th packet, her age decreases by $G_{i,\ell}$; in particular, the cumulative decrease from clearing $k$ packets is then $Z_i^{k}:=\sum_{\ell=1}^{k} G_{i,\ell}$. Another now familiar application of \Cref{prop:geom3} implies that with probability at least $1-A\exp(-Bw)$ (where $A,B$ are absolute constants depending only on $n,\epsilon, \bm{\lambda}$, not $w$), we have
\begin{equation*}
    \bigg\vert Z_i^{k}- \frac{k}{\lambda_i}\bigg\vert \leq \frac{(f_1\cdot \left(\sum_{i\in S_1} \lambda_i\right)\cdot \epsilon/2)\cdot w}{\lambda_i}.
\end{equation*}

Because the expected number of packets cleared by the queues in $S_1$ is at most $\lambda(S_1)\cdot f_1$ by definition, another familiar application of the Azuma-Hoeffding inequality also implies with probability at least $1-A'\exp(-B'w)$ (where $A',B'$ do not depend on $w$) that 
\begin{equation*}
    \sum_{t=1}^w \sum_{i\in S_1} X_{i,t}\leq (1+\epsilon/2)\cdot \left(\sum_{i\in S_1} \lambda_i\right)\cdot f_1\cdot w.
\end{equation*}
Combining these two estimates, we recall that
\begin{equation*}
    T^i_w=\max\left\{0,w-\sum_{t=1}^{J_i} G_{i,\ell}\right\} \geq w-\sum_{t=1}^{J_i} G_{i,\ell},
\end{equation*}
where $J_i=\sum_{t=1}^w X_{i,t}$. Thus, it follows under these good events that
\begin{align*}
    \sum_{i\in S_1} \lambda_i T^i_w&\geq w\cdot \left(\sum_{i\in S_1} \lambda_i\right) -\sum_{i\in S_1} \lambda_i \left(\sum_{\ell=1}^{J_i} G_{i,\ell}\right)\\
    &\geq w\cdot \left(\sum_{i\in S_1} \lambda_i\right)-\sum_{i\in S_1} \lambda_i\cdot \left(\frac{J_i}{\lambda_i}+\frac{(f_1\cdot \left(\sum_{i\in S_1} \lambda_i\right)\cdot \epsilon/2)\cdot w}{\lambda_i}\right)\\
    &= w\cdot \left(\sum_{i\in S_1} \lambda_i\right)-\sum_{i\in S_1} \left(J_i +\left(f_1\cdot \left(\sum_{i\in S_1} \lambda_i\right)\cdot \epsilon/2\right)\cdot w\right)\\
    &\geq w\cdot \left(\sum_{i\in S_1} \lambda_i\right)-w\cdot f_1 \cdot \left(\sum_{i\in S_1} \lambda_i\right)\cdot \epsilon/2-\sum_{t=1}^w \sum_{i\in S_1} X_{i,t}\\
    &\geq  w\cdot \left(\sum_{i\in S_1} \lambda_i\right)-w\cdot f_1 \left(\sum_{i\in S_1} \lambda_i\right)\epsilon/2-w\cdot(1+\epsilon/2)\cdot \left(\sum_{i\in S_1} \lambda_i\right)\cdot f_1\\
    &= w\cdot \left(\sum_{i\in S_1} \lambda_i\right) \cdot (1-(1+\epsilon)\cdot f_1).
\end{align*}
By a union bound, we thus find this occurs with probability $1-A\exp(-Bw)$ for some absolute possibly different constants $A,B>0$ that do not depend on $w$, thus concluding the proof of the second statement.

The proof for $s\geq 1$ is very similar, just with one extra condition. Suppose at time $t$, we have

\begin{equation*}
    \min_{i\in U_s} T^i_t -\max_{i\in S_{s+1}} T^i_t\geq  2\cdot \frac{w}{\lambda_n}.
\end{equation*}
Consider now the next $w$ steps, and let $G_{i,\ell}$ be as above, with $\ell=1$ to $w$ being the same geometric random variables, with $Z_{i}^k$ being the $k$th partial sum. Another familiar application of \Cref{prop:geom3} implies that with probability at least $1-A_1\exp(-A_2w)$, we have
\begin{equation*}
    Z_i^w < 1.5\cdot \frac{w}{\lambda_n}, \quad \forall i\in U_{s}.
\end{equation*}
Moreover, this event occurring implies that no queue in $U_s$ can possibly become younger than the any queue in $S_{s+1}$ on the next $w$ steps, even if they clear a packet in every step by our assumption. Therefore, on this event, we can apply the same analysis with the variables $X_{i,t}$ as above, just noting that conditioned on this occurring,
\begin{equation*}
    \mathbb{E}\left[\sum_{i\in S_{s+1}} X_{i,t+\ell}\bigg\vert \mathcal{F}_{t+\ell-1}\right]\leq \left(\sum_{i\in S_{s+1}} \lambda_i\right)\cdot f_{s+1},
\end{equation*}
as every queue in $U_s$ will have priority over every queue in $S_{s+1}$ on this window. An extremely similar argument via Azuma-Hoeffding and \Cref{prop:geom3} and taking a union bound so that the concentration of queues in $U_s$ above implies that the desired result holds with probability at least $1-A\exp(-Bw)$ for some constants $A,B>0$ not depending on $w$ (but again, on $n,\epsilon,\bm{\lambda}, \bm{\mu},\mathbf{p}$).

\end{proof}

While these above results characterize the long-run rates up to linear terms, we can provide slightly better control than offered by \Cref{thm:gamemain} for the groups such that $f_{k}>1$:
\begin{thm}[\Cref{thm:gamess}, restated]
Fix $\mathbf{p}$ and suppose that for some group $S_k$ output by the algorithm, $f_k>1$, so that $1-f_k<0$. Then, for each $i\in S_k$, $T^i_t$ is strongly stable.
\end{thm}
\begin{proof}
We just sketch the proof. It suffices to show this for the random variable $\max_{i\notin U_{k-1}}T^i_t$. Let $\epsilon>0$ be small enough such that $(1-(1-\epsilon)\cdot f_{k})<\eta<0$. Then let $w=a\cdot \lceil \frac{6}{\epsilon}\rceil^{n-1}$ be large enough so, on the event that $\max_{i\notin U_{k-1}}T^i_{\ell\cdot w}\geq f_{k}\cdot w$, then 
\begin{equation*}
    \mathbb{E}\left[\max_{i\notin U_{k-1}}T^i_{(\ell+1)\cdot w}-\max_{i\notin U_{k-1}}T^i_{\ell\cdot w} \bigg\vert \mathcal{F}_{\ell\cdot w}\right]<\beta<0
\end{equation*}
for some $\beta<0$, where $\mathcal{F}_{\ell\cdot w}$ is the natural filtration of the system; this can be done by \Cref{prop:upperbound}, noting that on the event where the proposition fails, the queue age can increase at most by $w$, and this can be drowned out in the expectation by the exponential decay of the probability bound by taking $w$ large enough. This proves the negative drift condition for the random process $Y_{\ell}:=\max_{i\notin U_{k-1}}T^i_{\ell\cdot w}$ with threshold value $f_{k}\cdot w$.

Then, for any even $p\geq 0$, $\mathbb{E}[\vert \max_{i\notin U_{k-1}}T^i_{(\ell+1)\cdot w}-\max_{i\notin T_{j-1}}T^i_{\ell\cdot w}\vert^p \vert \mathcal{F}_{\ell\cdot w}]$ is bounded by some constant $C_p>0$ for each $p$ depending only on $n,w, \bm{\lambda}$. This is because the difference between these is crudely upper bounded as random variables by a sum of at most $n\cdot w$ geometric random variables in the case that queues somehow clear a packet every round, which are easily seen to have bounded moments. By \Cref{thm:pemantle} (from \cite{pemantle1999moment}), this implies stochastic stability as the $p$th moment condition holds for arbitrarily large $p$.
\end{proof}
\end{document}